\documentclass[runningheads]{llncs}
\usepackage[T1]{fontenc}
\usepackage{graphicx}
\usepackage{hyperref} 
\usepackage{color}

\usepackage{wrapfig}

\usepackage[textsize=small]{todonotes}
\newcounter{todocounter}

\usepackage{amsmath,amssymb,cleveref}

\usepackage{thm-restate}

\usepackage{tikz}
\usetikzlibrary{automata}
\usetikzlibrary{snakes}

\usepackage{algorithm}
\usepackage[noend]{algpseudocode}

\newcommand\N{\mathbb N}

\newcommand\loc{\mathrm{loc}}
\newcommand\events{\mathcal E}
\newcommand\Foata{\mathsf F}
\newcommand\Foatalength[1]{|#1|_{\mathsf F}}
\newcommand\view{\mathrm{view}}
\newcommand\jointview{\overleftarrow{\view}}
\newcommand\incl{\subseteq}
\newcommand\diam{\mathsf{Diam}}
\newcommand\blue[1]{\textcolor{blue}{#1}}
\newcommand\A{\mathcal A}
\newcommand\B{\mathcal B}

\newcommand\length[1]{|#1|}
\newcommand\f{f}

\newcommand{\oo}{o}
\newcommand{\BB}{\mathfrak{B}}
\newcommand{\inn}{\mathsf{in}}
\newcommand{\out}{\mathsf{out}}
\newcommand{\restrict}[2]{#1{{\downarrow}_{#2}}}
\newcommand{\trest}{\restrict{t}{\Sigma^{\inn}(B)}}

\newcommand{\synchronise}{\mathsf{synchronise}}
\newcommand{\expand}{\mathsf{expand}}
\newcommand{\cut}{\mathsf{cut}}
\newcommand{\gcomm}{G_{(\Sigma, \loc)}}
\newcommand{\parent}{\mathsf{parent}}
\newcommand{\state}{\textsc{State}}
\newcommand{\cstate}{\textsc{CState}}
\newcommand{\Bacy}{\B_{\mathsf{acyc}}}
\newcommand{\Bcomb}{\B^{\mathsf{mod}}_{\mathsf{acyc}}}

\newcommand{\roo}{\mathsf{root}}
\newcommand\Lt{L_{\mathsf{tr}}}

\begin{document}

\title{Synthesising Asynchronous Automata from Fair Specifications}

\author{B\'eatrice B\'erard \inst{1}\and 
Benjamin Monmege \inst{2}\and 
B Srivathsan \inst{3,4}\and 
Arnab Sur\inst{3}}

\institute{
  Sorbonne Universit\'e, CNRS, LIP6, F-75005 Paris, France \\ \email{beatrice.berard@lip6.fr} \and
  Aix-Marseille Univ, CNRS, LIS, Marseille, France \\ \email{benjamin.monmege@univ-amu.fr} \and
  Chennai Mathematical Institute, Chennai, India  \\ \email{\{sri, arnabs\}@cmi.ac.in} \and
  CNRS IRL 2000, ReLaX, Chennai, India 
}


\maketitle

\begin{abstract}
Asynchronous automata are a model of distributed finite state processes synchronising on shared actions. A celebrated result by Zielonka shows how a deterministic asynchronous automaton (AA) can be synthesised, starting from two inputs: a global specification given as a deterministic finite-state automaton (DFA) and a distribution of the alphabet into local alphabets for each process.
The DFA to AA translation is particularly complex and has been revisited several times, with no complete prototype tool provided for the full construction.
In this work, we revisit this construction on a restricted class of
``fair'' specifications: a DFA describes a fair specification if in
every loop, all processes participate in at least one action — so, no
process is starved. For fair specifications, we present a new
construction to synthesise an AA. Our construction results in an AA
where every process has a number of local states that is linear in the
number of states of the DFA, and where the only exponential explosion
is related to a fairness parameter: the length of the longest word
that can be read in the DFA in which not every process
participates. We have implemented a prototype tool showing how it can
be applied to some examples, in particular a concrete one: the dining
philosophers problem. Finally, we show how this
construction can be combined with an existing construction for
hierarchical process architectures, in order to relax the fairness
assumption.
\end{abstract}




\section{Introduction}

Asynchronous automata (AA) are a foundational model for distributed systems, where independent processes synchronise on shared actions. These models enable the design and analysis of systems with decentralised control, making them critical in theoretical computer science and applications like parallel computing and distributed algorithms. Figure~\ref{fig:intro-aa-example} gives an example (from~\cite{Madhavan-Notes-Zielonka}) of an AA.
\begin{figure}[h]
\centering
\scalebox{0.88}{\begin{tikzpicture}[state/.style={draw, circle, thick, inner
    sep=2pt}]

    \begin{scope}[every node/.style={state}]
        \node (0) at (0,0) {\scriptsize $0$};
        \node (1) at (1.5, 0) {\scriptsize $1$};
        \node (2) at (3, 0) {\scriptsize $2$};
     \end{scope}

     \begin{scope}[->, >=stealth, thick, auto]
        \draw (-0.7, 0) to (0);
        \draw (0) to node {\scriptsize $a$} (1);
        \draw (1) to node {\scriptsize $a$} (2);
        \draw [red, dashed] (1) to [bend left=20] node [below] {\scriptsize $c$} (0);
        \draw [blue] (2) to [bend left=45] node [below] {\scriptsize $c$} (0);
     \end{scope} 
     \node at (1.5, -1.3) {\scriptsize $p_1$};

     \begin{scope}[xshift = 4.5cm]
        \begin{scope}[every node/.style={state}]
            \node (0) at (0,0) {\scriptsize $0$};
            \node (1) at (1.5, 0) {\scriptsize $1$};
            \node (2) at (3, 0) {\scriptsize $2$};
         \end{scope}
    
         \begin{scope}[->, >=stealth, thick, auto]
            \draw (-0.7, 0) to (0);
            \draw (0) to node {\scriptsize $b$} (1);
            \draw (1) to node {\scriptsize $b$} (2);
            \draw [red, dashed] (1) to [bend left=20] node [below] {\scriptsize $c$} (0);
            \draw [blue] (2) to [bend left=45] node [below] {\scriptsize $c$} (0);
         \end{scope} 
         \node at (1.5, -1.3) {\scriptsize $p_2$};
     \end{scope} 

     \draw (8, -1.5) to (8, 1.5);
     \begin{scope}[xshift=9cm]
        \node (00) at (0,1) {\scriptsize $0,0$};
        \node (01) at (1.5, 1) {\scriptsize $0,1$};
        \node (02) at (3, 1) {\scriptsize $0,2$};

        \node (10) at (0,0) {\scriptsize $1,0$};
        \node (11) at (1.5, 0) {\scriptsize $1,1$};
        \node (12) at (3, 0) {\scriptsize $1,2$};

        \node (20) at (0,-1) {\scriptsize $2,0$};
        \node (21) at (1.5, -1) {\scriptsize $2,1$};
        \node (22) at (3, -1) {\scriptsize $2,2$};
        
        \begin{scope}[->,>=stealth, thick, auto]
        \draw (-0.7, 1) to (00);
        \draw (00) to node {\scriptsize $b$} (01);
        \draw (01) to node {\scriptsize $b$} (02);
        \draw (10) to node [below] {\scriptsize $b$} (11);
        \draw (11) to node [below left] {\scriptsize $b$} (12);
        \draw (20) to node [below] {\scriptsize $b$} (21);
        \draw (21) to node [below] {\scriptsize $b$} (22);

        \draw (00) to node [left] {\scriptsize $a$} (10);
        \draw (10) to node [left] {\scriptsize $a$} (20);
        \draw (01) to node[above right] {\scriptsize $a$} (11);
        \draw (11) to node {\scriptsize $a$} (21);
        \draw (02) to node {\scriptsize $a$} (12);
        \draw (12) to node {\scriptsize $a$} (22); 

        \draw [red, dashed] (11) to [bend left =20] node[above right,node distance=1mm,yshift=-1mm]{\scriptsize $c$} (00);
        \draw [blue] (22) to [bend right = 20] node[above right,node distance=1mm,yshift=-1mm]{\scriptsize $c$}(00);
        \end{scope} 
     \end{scope}
\end{tikzpicture}}
\caption{Left: An AA. Right: its semantics seen as a DFA}
\label{fig:intro-aa-example}
\end{figure}
It contains two processes $p_1$, $p_2$ which control letters $\{a, c\}$ and $\{b, c\}$ respectively: so $a$ is local to $p_1$, $b$ is local to $p_2$, and $c$ is shared by $p_1$ and~$p_2$. The key point is the synchronisation mechanism on shared actions. Here, the AA is defined in a way that the two red-dashed $c$ transitions can synchronise, as well as the two blue $c$ transitions. However, the red-dashed $c$ in one cannot be done together with the blue $c$ of the other. The behaviour of the AA is described as a deterministic finite-state automaton (DFA) on the right. The language of this AA is $(([ab] + [aabb]).c)^*$ where $[ab]$ stands for either $ab$ or $ba$, and $[aabb]$ denotes the set of words obtained by shuffling two $a$'s and two $b$'s. 

From an AA, it is straightforward to construct a language-equivalent DFA, as shown in the example. The reverse process --- that is, given a DFA (global specification), construct a language-equivalent AA (distributed implementation) --- is highly non-trivial. A cornerstone result is Zielonka's theorem, which provides a method for synthesising a (deterministic) AA from a given DFA, and an alphabet distribution among processes~\cite{Zie87}. However, 
Zielonka’s construction is known to be extremely complex and has been the subject of numerous refinements and optimisations (e.g.~\cite{CorMet93,MukSoh97,Genest06}), leading to an optimal Zielonka-type construction~\cite{GenGim10}. In all these constructions, each process keeps track of what it believes to be the latest ``information'' available to every other process. When there is a shared action, the processes share their local information, reconcile them and update their local states. This underlying idea is implemented using a \emph{gossip automaton} in~\cite{MukSoh97,Madhavan-Notes-Zielonka}, and so-called \emph{zones} in~\cite{Genest06,GenGim10}. This idea yields a number of local states that is exponential in the number of processes (whereas the construction can be kept polynomial in the number of states of the DFA). 


Several solutions have been considered in order to avoid the resolution of this \emph{gossip problem}. For instance, compositional methods have been used in order to construct non-deterministic AA~\cite{Bau11,pighizzini1993synthesis,baudru2006unfolding}, still requiring a number of states exponential in the number of processes.
Another solution, leading to a quadratic number of states, is to restrict the topology of communications to be acyclic, and with any letter synchronising at most two processes~\cite{KRISHNA2013109}. This special case has recently been extended to remove the restriction on the number of processes synchronising on each letter, and to allow for a reconfiguration of the topology along the executions~\cite{HauLeh24}.
Another recent work~\cite{AdsGas24} has proposed a logic-based route for Zielonka's theorem by going through a local, past-oriented fragment of propositional dynamic logic. The AA is obtained by a cascade product of localised AA, which essentially operate on a single process.

Our goal in this work is
to find a ``meaningful'' restriction on the DFA specifications which
can result in a conceptually simpler and more efficient AA synthesis,
that can be implemented in a prototype tool. Indeed, there are very
few attempts of implementations of the synthesis algorithm based on
Zielonka's theorem, partly due to its intricacy. The work
of~\cite{SteEsp03} proposes a notion of safe asynchronous automata and
states that the problem of synthesizing a safe AA whose language is
contained in a DFA specification is undecidable. They go on to study
syntactic conditions on DFAs which make them equivalent to safe AA and
then propose a method to find sub-automata inside the DFA
specification that satisfy these conditions. From such sub-automata,
safe AA are synthesized using a modified Zielonka's procedure, which
is implemented and experimented.
Another implementation has been proposed in \cite{AksDin13} for a
fragment of asynchronous automata that are called \emph{realistic},
that is both deadend-free and locally accepting. Both studies define
restrictions on the DFA specification which are orthogonal to those we
consider.


In this article, we revisit AA synthesis with a focus on \emph{fair}
specifications. As an example, consider
the dining philosophers problem, introduced in \cite{Dij71}. The
problem is modeled for $n$ philosophers (the processes), sitting
around a table with a single chopstick lying between two consecutive
philosophers. A philosopher's actions (picking up/putting down the
left/right chopstick) are shared actions, involving the philosopher
and its two neighbors. The description of the authorised sequence of
actions, as well as the objective for all philosophers to eat
(i.e. have both chopsticks at some point in the future, before
releasing them), can be given as a DFA where each state stores the
current status of each philosopher. One of the objectives is to design
a distributed algorithm to avoid deadlock situations (where nothing
else can be done), and the starvation of any of the philosophers. In
our setting, we want to design an AA where every philosopher eats
infinitely often, which is traditionally called \emph{fairness} in the
realm of infinite words (every process that wants to play infinitely
often should be scheduled infinitely often). We strengthen the
requirement to  \emph{bounded-fairness} (though we will drop the
\emph{bounded-} in the rest of the article) where we not only want
that every process plays infinitely often, but also frequently
enough. Scheduling policies like round-robin scheduling impose bounded-fairness naturally~(see Section~2 of \cite{10.1145/295656.295659}, for
instance). Similar bounds are also frequently used in timed systems
to require, for instance, bounded response times.

From a DFA perspective, this bound on fairness requires that in every
loop of the DFA, every process participates (and thus in the long run,
no process is starved). The DFA in Figure~\ref{fig:intro-aa-example}
is fair, since every loop is closed on~$c$, which is a global action
where both processes participate. We introduce a natural number $k$ to
measure the quality of fairness: a DFA is $k$-fair whenever every word
of length $k$ that can be read in the DFA makes every process
participate at least once. The DFA in
Figure~\ref{fig:intro-aa-example} is $3$-fair. Notice that this
parameter $k$ is a priori independent of the number of processes and
the size of the DFA: there are examples (see
Example~\ref{ex:independent} for one) of DFA with arbitrary size and
with a distribution of actions among arbitrarily many processes that
are $k$-fair for a fixed parameter~$k$.

After studying the fairness of trace languages in Section~\ref{sec:fair}, we introduce a novel construction that simplifies the synthesis process for fair specifications in Section~\ref{sec:contrib-fair}. 
A crucial technical ingredient of our construction is the notion of Foata Normal Form (FNF) \cite{CarFoa69}. 
Our construction
does not require the use of gossip automata, or other 
sophisticated tools used in other constructions.
As a result, our approach ensures that the resulting AA is linear in the size of the DFA, polynomial in the size of the alphabet, and the only exponential explosion is related to the parameter $k$. The complexity is independent of the number of processes. We can even extend our synthesis procedure to synthesise an AA recognising all $k$-fair traces of a (potentially non-fair) DFA specification: in the dining philosophers problem, this allows us to ensure that every philosopher is active from time to time. We present a tool implementing our synthesis procedures in Section~\ref{sec:implementation}. 

We finally weaken the fairness restriction in Section~\ref{sec:combination}, by showing that our new construction can be combined with results known for hierarchical process architectures (as studied in \cite{KRISHNA2013109,HauLeh24}). In these architectures, the set of processes is organised as a tree where each process can communicate only with its parent or its children in the tree. We show that each node in this tree can be enlarged into a \emph{bag} of processes. If the DFA restricted to the alphabet of each of these bags is fair, then our fairness construction can be coupled with the construction known for hierarchical architectures. As a possible application, we can imagine a pool of servers that are communicating in a hierarchical fashion, each linked to some clients where a server and its clients should interplay in a fair fashion (ignoring what happens outside). We once again obtain a construction of polynomial complexity for a fixed value of fairness parameter $k$.


Detailed proofs of all results are presented in the clearly marked appendices.

\section{Preliminaries}\label{sec:prelim}

Let $\Sigma$ be a finite alphabet. A \emph{concurrent alphabet} is a pair $(\Sigma, I)$ with $I\subseteq \Sigma\times \Sigma$ the \emph{independence relation}, being a symmetric and irreflexive relation. The corresponding dependence relation is the set $D = (\Sigma\times \Sigma) \setminus I$. 
A concrete independence relation can be obtained by distributing the letters into a finite set $P$ of processes, formally defined by a function $\loc\colon \Sigma \to 2^P$ where $\loc(a)$ is the set of processes that can read the letter $a \in \Sigma$. We further define $\Sigma_p = \{a \in \Sigma \mid p \in \loc(a)\}$ as the alphabet of process $p\in P$. We call $(\Sigma, \loc)$ a \emph{distributed alphabet}. Such a distribution naturally leads to an independence relation $I_\loc = \{(a, b) \in \Sigma\times \Sigma \mid \loc(a) \cap \loc(b) = \emptyset\}$.
We extend the function $\loc$ to words by letting $\loc(\varepsilon) = \emptyset$, and $\loc(ua) = \loc(u) \cup \loc(a)$ for all $u \in \Sigma^*$ and $a \in \Sigma$.

\begin{example}\label{ex:distributed-alphabet}
  Let $\Sigma = \{a, b, c, d\}$ and $P  = \{p_1, p_2, p_3\}$ be the processes. Consider the distribution $\loc(a) = \{p_1, p_2\}$, $\loc(b) = \{p_1, p_3\}$, $\loc(c) = \{p_2\}$, and $\loc(d) = \{p_3\}$. Then $\Sigma_{p_1} = \{a, b\}$, $\Sigma_{p_2} = \{a, c\}$ and $\Sigma_{p_3} = \{b, d\}$. The independence relation is $I_\loc = \{(a, d), (d, a), (d, c), (c, d), (c, b), (b, c)\}$.
\end{example}



\paragraph*{Traces.} A \emph{trace} over a concurrent alphabet $(\Sigma, I)$ is a labelled partial order $t = (\events, \leq, \lambda)$ where $\events$ is a set of \emph{events}, $\lambda\colon \events \to \Sigma$ labels each event by a letter, and $\leq$ is a partial order of $\events$ satisfying the following conditions: 
\begin{itemize}
  \item $(\lambda(e), \lambda(f)) \notin I$ implies $e\leq f$ or $f\leq e$;
  \item $e \lessdot f$ implies $(\lambda(e), \lambda(f)) \notin I$ where ${\lessdot} =  {<} \setminus {<^2}$ is the immediate successor relation induced by the partial order: $\{(e, f) \mid e < f \text{ and } \lnot \exists g.~ e < g <f \}$.
\end{itemize}
The number of events in the trace $t$ is denoted as $|t|$, and is called its \emph{length}.

A word $w=a_0\cdots a_{n-1}\in \Sigma^*$ gives rise to a unique trace by putting one event per position in the word, and defining the successor relation $\lessdot$ as all the pairs $(i, j)$ of positions such that $i<j$, $(a_i, a_j)\notin I$ and there are no positions $k$ such that $i<k<j$ and $(a_i, a_k), (a_k, a_j)\notin I$. We call this word a \emph{linearisation} of the trace. Two words $w$ and $w'$ mapped to the same trace are said to be \emph{equivalent}, denoted by $w\sim w'$. We write $[w]$ for the equivalence class of $w$ with respect to $\sim$. In the following, we use both representations (partial orders and equivalence classes) of traces interchangeably.

Minimal (respectively, maximal) elements of a trace $t$ are all the events that have no smaller (respectively, larger) events. The set of minimal (respectively, maximal) events of $t$ is denoted by $\min(t)$ (respectively, $\max(t)$). 
Given two traces $t_1 = (\events_1, \leq_1, \lambda_1)$ and $t_2 = (\events_2, \leq_2, \lambda_2)$, the concatenation $t_1 t_2$ is the trace $(\events', \leq', \lambda')$ where $\events'= \events_1 \cup \events_2$, $\lambda'(e) = \lambda_1(e)$ if $e \in \events_1$, and $\lambda'(e) = \lambda_2(e)$ otherwise, and $\leq'$ is ${\leq_1} \cup {\leq_2} \cup \{ (x, y) \mid  x \text{ is maximal in } t_1, y \text{ is minimal in } t_2, \text{ and } (x,y) \notin I\}$.  

\begin{example}\label{ex:trace} A trace over the distributed alphabet of Example~\ref{ex:distributed-alphabet} can be depicted 
  as follows, where the arrows between the events denote the relation $\lessdot$:
  \begin{center}
    \begin{tikzpicture}[>=latex,xscale = 1.5, yscale=1.5]
      \draw (0, 0) node {$a$};
      \draw (1, 0) node {$b$};
      \draw (2, 0) node {$d$};
      \draw (1, -0.5) node {$c$};
      \draw (2, -0.5) node {$a$};
      \draw (3, -0.5) node {$c$};
    
      \draw[->] (0.1, 0) -- (0.9, 0); 
      \draw[->] (0.1, -0.1) -- (0.9, -0.5); 
      \draw[->] (1.1, 0) -- (1.9, 0); 
      \draw[->] (1.1, -0.5) -- (1.9, -0.5); 
      \draw[->] (1.1, -0.1) -- (1.9, -0.4); 
      \draw[->] (2.1, -0.5) -- (2.9, -0.5); 
    \end{tikzpicture}
  \end{center}
   It is the trace associated with the equivalence class $[abcacd]$, that is equal, for instance, to the equivalence class $[acbdac]$. The minimal event is the one labelled by $a$ on the left. The maximal events are those labelled by $d$ and $c$ on the right. 
  
\end{example}



\paragraph*{Views.}
\newcommand\ideal[1]{#1{\downarrow}}

For a trace $t = (\mathcal E, {\leq}, \lambda)$, a subset $J\subseteq \mathcal E$ is called an \emph{ideal} of $t$ if for all $e\in J$, and $f\in \mathcal E$ such that $f\leq e$, we have $f\in J$. An ideal can also be seen as a trace by keeping the same partial order and labelling as in the original trace.
For a subset $X\subseteq \mathcal E$ of events, we let $\ideal{X}$ be the ideal $\{f\in \mathcal E \mid \exists x\in X \; f\leq x\}$. 
From a linearisation perspective, an ideal $s$ of a trace $t$ is related to a prefix of one of the linearisations: there must exist a linearisation $w$ of~$t$ that can be written as $w=uv$ where $s$ is the trace $[u]$. 


We identify special ideals called \emph{views}: the \emph{view} of a process $p$ in a trace $t$ is the ideal of~$t$ consisting of all events currently known by the process $p$. Formally, we let $\max_p(t)$ be the largest event of $t$ that is labelled in $\Sigma_p$. Then, the view of process $p$ in a trace $t$, denoted as $\view_p(t)$, is the ideal $\ideal{\max_p(t)}$. 
The view of a set of processes $X \subseteq P$ in a trace $t$ is the ideal 
$\ideal{\{\max_p(t) \mid p\in X\}}$, obtained as the union of the views of the processes in $X$. 

\begin{example}\label{ex:views} Consider again the trace $t = [abcacd]$ presented in Example~\ref{ex:trace}. Then, $\view_{p_1}(t) = [abca]$, $\view_{p_2}(t) = [abcac]$, $\view_{p_3}(t) = [abd]$, and $\view_{\{p_1, p_3\}}(t) = [abcad]$. 
\end{example}


\paragraph*{Foata normal form.} The Foata normal form (FNF) of a trace encodes a maximal parallel execution of the trace \cite{CarFoa69}. To define this notion formally, we use the concept of \emph{steps} from~\cite{Diekert-Muscholl-11}. A \emph{step} is a non-empty subset $S \incl \Sigma$ of pairwise independent letters: it is sometimes called a \emph{clique}, from the point of view of the graph of the independence relation. In a step, all letters can be executed in parallel.  Observe that the set of labels of the minimal elements of a trace $t$ provides a maximal step for $t$. Then, the FNF $\Foata(t)$ of a trace $t$ is a sequence of steps $\varphi = S_1 S_2 \cdots S_m$ where $S_1$ is the set of minimal elements of $t$, and $S_2 \cdots S_m = \Foata(t')$ where $t'$ is the trace obtained by removing all minimal elements from $t$. The FNF is a unique decomposition of the trace into maximal steps. 

\begin{example}\label{ex:Foata} The FNF of the trace $[abcacd]$ from Example~\ref{ex:trace} is $\{a\}\{b,c\}\{a,d\}\{c\}$.
\end{example}

We denote by $\varphi_i$ the $i$-th step of the FNF $\varphi$. We shall also denote by $\Foatalength t$, the number of steps in the FNF decomposition of the trace $t$, and we call it the \emph{Foata length of $t$}.
In order to simplify further explanations, we suppose that a step $\varphi_i$ exists as the empty set, for all $i$ greater than the length of $\varphi$: we do not depict this infinite sequence of empty sets when we give examples of FNF.

Interestingly, FNF are increasing with respect to ideals: if a trace $s$ is an
ideal of~$t$, there is an injective correspondence from steps of $s$ to the
$\Foatalength s$ first steps of $t$. Essentially, the FNF $\Foata(s)$ is
obtained by deleting from $\Foata(t)$ the events not in~$s$. 

\begin{restatable}{lemma}{prefixFoata}
  \label{lem:prefixFoata} 
  Let $t$ be a trace and $s$ an ideal of $t$. Then for all $i \leq \Foatalength s$, $\Foata(s)_i \subseteq \Foata(t)_i$.
\end{restatable}

\begin{example}
  Consider again the trace $t = [abcacd]$ of Example~\ref{ex:trace}. The FNF of the ideal $s = [abd]$ is $\{a\}\{b\}\{d\}$. 
   We can extend the ideal $s$ by the remainder of the the trace $t$, $[cac]$ letter by letter. The first letter $c$ commutes with $b$ and $d$ but not with $a$ and thus must belong to the second step. Arguing similarly we see that the next letter $a$ must belong to the third step and the final letter $c$ must be appended in a new step. We thus end up with the FNF of $t$ as expected.
\end{example}

Even more interestingly, if we know the FNF of the views of two (or more) processes $p_1$ and $p_2$, it is possible to deduce the FNF of the view of $\{p_1, p_2\}$: they can accumulate their knowledge simply by taking pairwise unions of the steps. This gives us an algorithm to compute $\view_X(t)$ for some $X \subseteq P$, given $\view_p(t)$ for all $p \in X$. 
In the following, given two FNF $\varphi$ and $\varphi'$, we write $\varphi \cup \varphi'$ for the sequence of steps $(\varphi_1 \cup \varphi'_1) \cdots (\varphi_m \cup \varphi'_m)$ where $m$ is the maximal length of $\varphi$ and $\varphi'$ (remember that we have added empty steps at the end of the FNF so that $\varphi_k$ has a meaning, even for $k$ greater than the length of $\varphi$).

\begin{restatable}{lemma}{viewFoata}
  \label{lem:viewFoata}
  Let $t$ be a trace and $X \subseteq P$. Then, $\Foata(\view_X(t)) = \bigcup\limits_{p\in X} \Foata(\view_p(t))$.
\end{restatable}

\begin{example}
  Continuing Example~\ref{ex:views}, we have $\Foata(\view_{p_1}(t)) = \{a\}\{b,c\}\{a\}$ and $\Foata(\view_{p_3}(t)) = \{a\}\{b\}\{d\}$, from which we can indeed deduce $\Foata(\view_{\{p_1, p_3\}}(t)) = \{a\}\{b,c\}\{a,d\}$.
\end{example}

\paragraph*{Regular Trace-Closed Languages and Asynchronous Automata.}
A language $L \subseteq \Sigma^*$ is said to be \emph{trace-closed} (for the independence relation $I$) if for all $w, w' \in \Sigma^*$ such that $w \in L$ and $w' \sim w$, we have $w' \in L$.
A trace-closed language is \emph{regular} if it is accepted by a finite state automaton $\A = (Q, \Sigma, Q_0, \Delta, Q_f)$, where $Q$ is the set of states, $Q_0$ are the initial states, $Q_f$ the final ones and $\Delta\subseteq Q\times \Sigma\times Q$ are the transitions. We denote by $q\xrightarrow a q'$ the transition $(q, a, q') \in \Delta$. In the following, we let $\Delta(q, w)$ be the set of states $q'$ such that there is a sequence of transitions $q=q_0\xrightarrow {a_0} q_1 \xrightarrow {a_1} \cdots \xrightarrow{a_{n-1}} q_{n} = q'$, if $w=a_0 a_1 \ldots a_{n-1}$. As usual, the language recognised by $\A$ is the set of words $w$ such that there exists $q\in Q_0$ with $\Delta(q, w)\cap Q_f\neq \emptyset$.
Trace-closed regular languages are recognised by automata having a special syntactical property, the \emph{diamond property}. 

\begin{definition} For a concurrent alphabet $(\Sigma, I)$, a finite state automaton $\A = (Q, \Sigma, Q_0, \Delta, Q_f)$ satisfies the \emph{diamond property} if for all $q, q', q'' \in Q$ and $(a, b) \in I$ such that $q \xrightarrow{a} q' \xrightarrow b q''$, there exists $q'''\in Q$ such that $q \xrightarrow{b} q''' \xrightarrow a q''$. 
\end{definition}

The diamond property is a sufficient condition for an automaton to recognise a trace-closed language. Indeed, if $\A$ is a finite state automaton satisfying the diamond property then for all $q\in Q$ and $w, w' \in \Sigma^*$ such that $w \sim w'$, we have $\Delta(q, w) = \Delta(q, w')$.
%
%
%
In particular, $\Delta(q, t)$ is well-defined even for traces $t$, since every linearisation of $t$ goes to the same state. In the following, we let $\Lt(\A)$ be the set of traces accepted by $\A$.

A finite state automaton $\A = (Q, \Sigma, Q_0, \Delta, Q_f)$ is said to be deterministic if $Q_0$ is a singleton, and for all $q\in Q$, and $a\in \Sigma$, $\Delta(q, a)$ is of cardinality at most $1$. We generally write $\A = (Q, \Sigma, q_0, \delta, Q_f)$ with $Q_0 = \{q_0\}$, and $\delta(q, a) = q'$ for all $(q, a, q')\in \Delta$.
Every finite state automaton satisfying the diamond property can be transformed into an equivalent deterministic finite state automaton (DFA) satisfying the diamond property. 



Zielonka's theorem aims at distributing a trace-closed regular language to the individual processes. This can be formally stated with (deterministic) asynchronous automata.

\begin{definition}
  An asynchronous automaton (AA) over the distributed alphabet $(\Sigma, \loc)$ is a tuple $\B = ((Q_p)_{p \in P}, \Sigma, q^0, (\delta_a)_{a \in \Sigma}, F)$ where
  \begin{itemize}
    \item $Q_p$ is the set of local states of a process $p \in P$,
    \item $\delta_a\colon \Pi_{p \in \loc(a)} Q_p \to \Pi_{p \in \loc(a)} Q_p$ is the transition function associated with $a \in \Sigma$,
    \item $q^0 \in \Pi_{p \in P} Q_p$ is the global initial state, and
    \item $F \subseteq \Pi_{p \in P} Q_p$ is the set of global final states.
  \end{itemize}
  The AA is said to be \emph{finite} if each process has a finite number of states.
  \label{def:asynchronousAutomata}
\end{definition}



The semantics of an AA is given by a DFA whose set of states is the product $\Pi_{p\in P} Q_p$. We call these states the \emph{global states} of $\B$. Its initial global state is $q^0$, and its final global states are given by $F$. Moreover, for each letter $a\in \Sigma$, we let $(q_p)_{p\in P} \xrightarrow a (q'_p)_{p\in P}$ if for all $p'\notin\loc(a)$, $q'_p = q_p$, and $\delta_a((q_p)_{p\in \loc(a)}) = (q'_p)_{p\in \loc(a)}$. We can show that this DFA satisfies the diamond property, and thus, we let $\Lt(\B)$ be its trace language, and say that the AA~$\B$ recognises $\Lt(\B)$.

\begin{theorem}[\cite{Zie87}]
  For every DFA $\A$ over the concurrent alphabet $(\Sigma, I)$ satisfying the diamond property, and every distributed alphabet $(\Sigma, \loc)$ such that $I_\loc = I$, there exists a finite AA $\B$ over $(\Sigma, \loc)$ such that $\Lt(\A) = \Lt(\B)$.
\end{theorem}

\section{Fair specifications}\label{sec:fair}

Fairness is an important condition in the verification of distributed systems, that imposes conditions on distributed runs such that no process starves in the long run. In our situation where finite-state automata and finite traces are considered, we strengthen the fairness condition to obtain a notion of \emph{bounded-fairness}: we require that no process lags behind the other ones more than a bounded number of steps. Formally, for a positive integer $k$, a trace $t$ is $k$-\emph{fair} if for every factor $u$ of any linearisation $w$ of $t$, that has length at least $k$, we have $\loc(u) = P$ (where $P$ is the set of all processes).

\begin{example}
  Consider the trace $[abcacd]$ over the distributed alphabet of Example~\ref{ex:trace}. Notice that in every factor of the linearisation $abcacd$ with length at least 4, all processes participate. This property can be verified to be true in every linearisation. Hence $[abcacd]$ is $4$-fair. However, $[abcacd]$ is not $3$-fair since the factor $cac$ in the linearisation $abcacd$ does not involve process~$p_3$.
\end{example}

As intended, the $k$-fairness of a trace ensures that views of different processes do not differ too much.

\begin{restatable}{lemma}{pairOfProcsAtmostKApart}
  \label{lem:pairOfProcsAtmostK-1Apart}
  Let $t$ be a $k$-fair trace. For all pairs of processes $p, p' \in P$, $|\length{\view_p(t)} - \length{\view_{p'}(t)}| \leq k-1$.
\end{restatable}


The situation is indeed even better. The view of a process, when written in FNF, coincides with the whole trace, except in the last steps that contain at least $k-1$ letters in total. To describe such a property more easily, we define a measure on traces: for a natural number $\ell$ and a trace $t$ of length at least $\ell$, we let $\f(t, \ell)$ be the largest natural number $i$ such that the steps $\Foata(t)_i \cdots \Foata(t)_{\Foatalength{t}}$ contain at least $\ell$ letters in total. Notice in particular that, by maximality of $i$, the steps $\Foata(t)_{i+1} \cdots \Foata(t)_{\Foatalength{t}}$ contain in total at most $\ell-1$ letters.

\begin{restatable}{lemma}{completeTraceBeyondKFoataComponents}
  \label{lem:completeTraceBeyondKFoataComponents}
    Let $t$ be a $k$-fair trace and $p \in P$ be a process with a view of length at least $k-1$. For all $i < \f(\view_p(t), k-1)$, $\Foata(t)_i = \Foata(\view_p(t))_i$.
\end{restatable}

\begin{example}
  The bound given in Lemma~\ref{lem:completeTraceBeyondKFoataComponents} is optimal in the following sense. Consider the alphabet $\{a, b\}$ with two processes $p_1$ and $p_2$ such that $\loc(a)= \{p_1\}$, $\loc(b) = \{p_2\}$. The trace $t$ whose linearisation is $b a^{k-1}$ is $k$-fair, and the view of process $p_1$ has linearisation $a^{k-1}$ of length $k-1$, with FNF where all steps are singletons. The FNF of $t$ is $\{a, b\} \{a\}^{k-2}$, and thus $p_1$ is not fully aware of the first step $\{a, b\}$. In particular, $\f(\view_{p_1}(t), k-1) = 1$, and indeed the first step is such that $\Foata(t)_1 \neq \Foata(\view_{p_1}(t))_1$.
\end{example}

We need a slightly more refined result in the following, since the full trace is known by no process, and thus we need to understand for a process $p$, what guarantee it has on the view of another process $p'$, to enable a successful synchronisation between them.

\begin{restatable}{lemma}{everyoneKnowsBeyondKFoataComponents}
  \label{lem:everyoneKnowsBeyond2KFoataComponents}
  Let $t$ be a $k$-fair trace and $p \in P$ be a process with a view of length at least $2k-2$. Then for all $i < \f(\view_p(t), 2k-2)$ and for all $p' \in P$, $\Foata(t)_i = \Foata(\view_{p'}(t))_i$.
\end{restatable}

Thus, for a $k$-fair trace $t$, the views of two different processes $\view_p(t)$ and $\view_{p'}(t)$, when considered in their FNF, have an identical prefix, and differ only in the last few steps. 
The identical prefix in fact corresponds to the given full trace $t$. In our construction of an AA, we use this fact crucially to maintain only a finite suffix of the trace in the local states.


\begin{example}
  Once again, the result of Lemma~\ref{lem:everyoneKnowsBeyond2KFoataComponents} is optimal. Consider an extension of the previous example where $\Sigma=\{a, b, c, d\}$, $\loc(a) = \{p_1\}$, $\loc(b)= \{p_2\}$, $\loc(c) = \{p_1, p_2\}$, and $\loc(d) = \{p_1, p_3\}$. Consider the trace $t = b a^{k-2} d c a^{k-2}$, whose FNF is $\{a, b\} \{a\}^{k-3} \{d\}\{c\}\{a\}^{k-2}$. The views of the 3 processes are: $\view_{p_1}(t) = t$, $\view_{p_2}(t) = b a^{k-2} d c$, and $\view_{p_3}(t)= a^{k-2} d$. The trace $t$ is $k$-fair: indeed all factors of any linearisation of length $k$ either pick both the letters $d$ and $b$, or the letters $d$ and $c$. Notice that it is not $(k-1)$-fair, since it contains the factor $a^{k-2}d$ in which process $p_2$ does not participate. Finally, $p_1$ has a view of length $2k-1$, and $f(\view_{p_1}(t), 2k-2) = 1$. Indeed, $p_3$ is not fully aware of the first step of the FNF of $t$.
\end{example}

A DFA satisfying the diamond-property is said to be \emph{fair} if in every
loop, all processes participate in at least one action. Since this must be true
already on loops without repetition of states, we can refine this definition to
introduce a parameter $k$ of fairness, as before. We first do it on the
languages: a trace language $X$ is $k$-fair if for all $t \in X$, $t$ is
$k$-fair. We can also refine the definition of fairness of a DFA to take the
parameter $k$ into account.

\begin{definition} \label{def:kFairDFA}
  A DFA $\A = (Q, \Sigma, q_0, \delta, Q_f)$ satisfying the diamond property is said to be $k$-fair if for each state $q \in Q$ reachable from an initial state, and for every word $u \in \Sigma^*$ such that $|u| \ge k$ and $\delta(q, u)$ is co-reachable (i.e.~there is a path from $\delta(q, u)$ to a final state), we have $\loc(u) = P$.
\end{definition}

This definition is indeed a characterisation of $k$-fair trace languages:

\begin{restatable}{proposition}{fairDFALanguage}
  Let $\A$ be a DFA satisfying the diamond property. Then, $\A$ is $k$-fair if and only if the language $\Lt(\A)$ of traces is $k$-fair.  \label{prop:fairDFALanguage}
\end{restatable}

As a corollary, we see that for a language $L$ of traces, there exists $k$ such that it is $k$-fair if and only if in every loop of any DFA recognising $L$ (satisfying the diamond-property), all processes participate. This is thus easy (polynomial-time) to test if a given DFA recognises a fair language (i.e.~$k$-fair for some $k$).
It is also possible to test if a DFA is $k$-fair, for a given $k$, since it is enough to check the property of Definition~\ref{def:kFairDFA} for every word $u$ of length~$k$. This naive algorithm requires an exponential-time complexity with respect to~$k$. In~\Cref{app:deciding-kfair}., we propose a dynamic-programming algorithm that solves this problem in polynomial-time.

We end this section with an example of a family of languages where the alphabet size, the required number of states in a DFA, and the number of processes increase arbitrarily, while $k$ remains constant. 
This shows that the fairness parameter can be helpful to describe the complexity of a trace language, in a different way than other parameters. This is in particular motivating for the contribution we present in the next section, where the only exponential component in the complexity depends on the parameter $k$.

\begin{example}\label{ex:independent}
  Let $\Sigma = \{ c, a_1, a_2, \dots, a_n \}$, $P = \{p_1, \dots, p_n\}$ and $\Sigma_{p_i} = \{a_i, c\}$ for all $i$. Consider the trace language described by $L_n = ((\bigcup_{1 \le i < j \le n} [a_i a_j]).[c] )^*$: it is the alternation of two among the $n$ letters $\{a_1, \ldots, a_n\}$ (where two distinct processes participate concurrently), and the global synchronisation letter $c$ (where all processes participate).
  Hence, $L_n$ is $3$-fair for all $n$, though it requires a DFA with $\Omega(n)$ states. 
\end{example}

\section{A Zielonka's theorem for fair trace languages}\label{sec:contrib-fair}

Our first contribution is a new, and more efficient, proof of Zielonka's theorem in the specific case where the specification trace language is fair:


\begin{theorem}\label{thm:fair-AA}
  For every $k$-fair DFA $\A$ over the distributed alphabet $(\Sigma, \loc)$ satisfying the diamond property,
  there exists a finite AA $\B$ over $(\Sigma, \loc)$ such that $\Lt(\A)=\Lt(\B)$, where every set of local states (for each process) is of size bounded by $O(n \times k\times |\Sigma|^{3k-3})$, where $n$ is the number of states of $\A$. 
\end{theorem}
  
Notice that the complexity is linear in the size of the DFA, and polynomial in the size of the alphabet --- with the degree of the polynomial only depending linearly in the fairness parameter. 
We now describe an overview of the construction over an example. A complete  proof of Theorem~\ref{thm:fair-AA} is provided in~\Cref{app:full-proof-fairness}.

This construction starts from $\A$ and proceeds in three steps, each building an AA accepting $\Lt(\A)$.
In the first step, we build an infinite AA whose states are tuples of the FNF of process views. In the second step, we explain how to cut these views, only keeping suitable suffixes but adding local states and unbounded counters. The third and final AA is then obtained by bounding the counter values using modulo counting.


\begin{figure}[btp]
  \centering
  \begin{tikzpicture}[state/.style={draw, circle, thick, inner
      sep=2pt}]
    \begin{scope}[every node/.style={state}]
      \node (0) at (0,0) {$q_0$};
      \node (1) at (2,0) {$q_1$};
      \node (2) at (0,-2) {$q_2$};
      \node [double] (3) at (2,-2) {$q_3$};
    \end{scope}

    \begin{scope}[->, >=stealth, thick, auto]
      \draw (-0.7, 0) to (0);
      \draw (0) to node {\scriptsize $c$} (1);
      \draw (0) to node [left] {\scriptsize $d$} (2);
      \draw (0) to [bend left=10] node {\scriptsize $a$} (3);
      \draw (1) to node {\scriptsize $d$} (3);
      \draw (2) to node [below] {\scriptsize $c$} (3);
      \draw (3) to [bend left=10] node {\scriptsize $b$} (0);
    \end{scope}

     \draw [gray, thin] (3.5, -2.2) to (3.5, 0.1);
      



    \begin{scope}[xshift = 5cm]
      \node [left] at (1, 0) {\scriptsize Word $w: $};
      \node [right] at (1, 0) {\scriptsize $abcdbdcb$};
      \node [left] at (1,-1) {\scriptsize Trace $[w]:$};
      \node [left] at (1, -2) {\scriptsize $\Foata([w]): $};
      \node [right] at (1, -2) {\scriptsize $\{a\}~~\{b\}~~\{c,d\}~~\{b\}
        ~~\{c, d\}~\{b \}$};
      
      \node (a) at (1.3, -0.5) {\scriptsize $a$};
      \node (b1) at (2, -0.5) {\scriptsize $b$};
      \node (d1) at (2.7, -0.5) {\scriptsize $d$};
      \node (b2) at (3.4, -0.5) {\scriptsize $b$};
      \node (d2) at (4.1, -0.5) {\scriptsize $d$};
      \node (c1) at (2.7, -1.5) {\scriptsize $c$};
      \node (c2) at (4.1, -1.5) {\scriptsize $c$};
      \node (b3) at (4.8, -0.5) {\scriptsize $b$};

      \begin{scope}[->, >=stealth, thin]
        \draw (a) to (b1);
        \draw (b1) to (d1);
        \draw (d1) to (b2);
        \draw (b2) to (d2);
        \draw (d2) to (b3);
        \draw (b1) to (c1);
        \draw (b2) to (c2);
        \draw (c1) to (b2);
        \draw (c2) to (b3);
      \end{scope}

    \end{scope}
  \end{tikzpicture}
  \caption{(Left) A DFA satisfying the diamond property; 
  (Right) a word $w$, the trace $[w]$ and its Foata normal form
  $\Foata([w])$.}
\label{fig:running-eg-intro}
\end{figure}

Consider the DFA specification $\A$ in Figure~\ref{fig:running-eg-intro} over the distributed alphabet $(\{a, b, c, d\}, \loc)$ with $\loc(a) = \{p_1, p_2\}$, $\loc(b) = \{p_1, p_3\}$, $\loc(c) = \{p_2, p_3\}$, and $\loc(d) = \{p_1\}$ (note that the distributed alphabet is different from the one previously used in our examples). The associated independence relation is $I_\loc = \{(c, d), (d, c)\}$. The word $abcdbdcb$ can be read from state $0$, reaching back to state $0$. It is a linearisation of the trace depicted on the right of the figure, which also gives its FNF. 

\begin{figure}
  \centering

\scalebox{.76}{\begin{tikzpicture}

    \draw [gray] (0,0) to (15.2, 0);
    \draw [gray] (0,0.5) to (15.2, 0.5);
    \draw [gray] (0,1) to (15.2, 1);
    \draw [gray] (0,1.5) to (15.2, 1.5);

    \node [left] at (0, 1.25) {\scriptsize $p_1$};
    \node [left] at (0, 0.75) {\scriptsize $p_2$};
    \node [left] at (0, 0.25) {\scriptsize $p_3$};
    
    \draw [gray] (0, 0) to (0, 1.5);
    \draw [gray] (1, 0) to (1, 1.5);
    \node at (0.5, 1.75) {\scriptsize $a$};
    \node at (0.5, 1.25) {\scriptsize $\{\textcolor{blue}{ a }\}$};
    \node at (0.5, 0.75) {\scriptsize $\{\textcolor{blue}{ a }\}$};
    \node at (0.5, 0.25) {\scriptsize $\{ \}$};

    \draw [gray] (2, 0) to (2, 1.5);
    \node at (1.5, 1.75) {\scriptsize $b$};
    \node at (1.5, 1.25) {\scriptsize $\{a\} \{ \textcolor{blue}{ b }\}$};
    \node at (1.5, 0.75) {\scriptsize $\{a\}$};
    \node at (1.5, 0.25) {\scriptsize $\{ \textcolor{blue}{ a }\}\{\textcolor{blue}{ b}\}$};

    \draw [gray] (3.5, 0) to (3.5, 1.5);
    \node at (2.75, 1.75) {\scriptsize $c$};
    \node at (2.75, 1.25) {\scriptsize $\{a\} \{b\}$};
    \node at (2.75, 0.75) {\scriptsize $\{a\} \{ \textcolor{blue}{ b }\}
        \{ \textcolor{blue}{ c }\}$};
    \node at (2.75, 0.25) {\scriptsize
      $\{a\}\{b\}\{ \textcolor{blue}{ c }\}$};

    \draw [gray] (5,0) to (5, 1.5);
    \node at (4.25, 1.75) {\scriptsize $d$};
    \node at (4.25, 1.25) {\scriptsize $\{a\} \{b\} \{
      \textcolor{blue}{ d }\}$};
    \node at (4.25, 0.75) {\scriptsize $\{a\}\{b\}
        \{c\}$};
    \node at (4.25, 0.25) {\scriptsize
      $\{a\}\{b\}\{c\}$};

    \draw [gray] (7,0) to (7, 1.5);
    \node at (6, 1.75) {\scriptsize $b$};
    \node at (6, 1.25) {\scriptsize $\{a\} \{b\} \{
      \textcolor{blue}{c}, d\} \{\textcolor{blue}{b}\}$};
    \node at (6, 0.75) {\scriptsize $\{a\}\{b\}\{c\}$};
    \node at (6, 0.25) {\scriptsize $\{a\}\{b\}\{c,
      \textcolor{blue}{d}\}\{\textcolor{blue}{b}\}$}; 

    \draw [gray] (9.5, 0) to (9.5, 1.5);
    \node at (8.25, 1.75) {\scriptsize $d$};
    \node at (8.25, 1.25) {\scriptsize
      $\{a\}\{b\}\{c,d\}\{b\}\{\textcolor{blue}{d}\}$};
    \node at (8.25, 0.75) {\scriptsize $\{a\} \{b\} \{c\}$};
    \node at (8.25, 0.25) {\scriptsize $\{a\} \{b\} \{c, d\} \{b\}$};

    \draw [gray] (12, 0) to (12, 1.5);
    \node at (10.75, 1.75) {\scriptsize $c$};
    \node at (10.75, 1.25) {\scriptsize $\{a\} \{b\} \{c,d\}\{b\}
      \{d\}$};
    \node at (10.75, 0.75) {\scriptsize $\{a\} \{b\} \{c,
      \textcolor{blue}{d}\} \{\textcolor{blue}{b}\}
      \{\textcolor{blue}{c}\}$};
    \node at (10.75, 0.25) {\scriptsize $\{a\} \{b\} \{c, d\} \{b\} \{
      \textcolor{blue}{c}\}$};

    \draw [gray] (15.2, 0) to (15.2, 1.5);
    \node at (13.6, 1.75) {\scriptsize $b$};
    \node at (13.6, 1.25) {\scriptsize $\{a\} \{b\} \{c,d\}\{b\}
      \{\textcolor{blue}{c}, d\} \{ \textcolor{blue}{b}\}$ };
    \node at (13.6, 0.75) {\scriptsize $\{a\}\{b\}\{c,d\}\{b\}\{c\}$};
    \node at (13.6, 0.25) {\scriptsize $\{a\}\{b\}\{c,d\}\{b\}\{c,
      \textcolor{blue}{d}\}\{\textcolor{blue}{b}\}$}; 
  \end{tikzpicture}}
  \caption{Run of the infinite asynchronous automaton corresponding to
    the finite state automaton of Figure~\ref{fig:running-eg-intro}, on the word
    $abcdbdcb$}
  \label{fig:run-of-infinite-aa}
\end{figure}

\smallskip

\noindent \textbf{Step 1.} Even without the fairness assumption, it is always possible to define an \emph{infinite} AA recognising $\Lt(\A)$, in which each process keeps its \emph{view of the current trace} as its local state. Furthermore, maintaining the view in the FNF allows to compute the synchronised views easily during shared actions. We illustrate this idea in Figure~\ref{fig:run-of-infinite-aa}. 
The beginning of the run of this infinite AA is depicted in Figure~\ref{fig:run-of-infinite-aa}. Initially, the views of all the processes are empty. When the letter $a$ is read, processes $p_1$ and $p_2$ update their view, while process $p_3$ is not aware of it. When the letter $b$ is read next, processes $p_1$ and $p_3$ synchronise. This implies that $p_3$ catches up by learning about the letter $a$ read before. Since $a$ and $b$ are dependent, we add a new step in the FNF of both views. We then read letters $c$ and $d$ in the same manner (note that we would obtain the same result by reading first $d$ and then $c$). When letter $b$ is read next, processes $p_1$ and $p_3$ learn from each other about the previous letters $c$ and $d$. They merge their views by making stepwise unions of their FNF (as justified by Lemma~\ref{lem:viewFoata}), before adding a new step with letter $b$. The run goes on like this, increasing the number of steps in the FNF. At the end, the acceptance status of the current run is obtained by once again merging the FNF of all processes (through a stepwise union), getting the FNF of the actual full trace, over which we can simply run the DFA to know whether the last state is final or not. Our objective now is to make this construction finite using the fairness assumption.

\smallskip

\noindent \textbf{Step 2.} The DFA of Figure~\ref{fig:running-eg-intro} is $4$-fair, since every sequence of 4 transitions makes every process participate: notice that it is not $3$-fair because of the word $dbd$ that can be read from state $q_1$, where process $p_2$ does not participate. The crux of our construction is Lemma~\ref{lem:everyoneKnowsBeyond2KFoataComponents}. As an example, let $t = abcdbdcb$, and consider $\view_{p_1}(t)$ as shown in the last column of Figure~\ref{fig:run-of-infinite-aa}. We have $k = 4$, and hence $2k - 2 = 6$, and $j = f(\view_{p_1}(t), 2k-2) = 3$. Observe that the prefix up to $j-1$ (given by $\{a\}\{b\}$ in this case) is known to every other process. Hence the process does not need to maintain this prefix in the local state. We can instead cut out these first steps of the FNF, and update the state of the DFA.

\begin{figure}[tbp]
\centering

\scalebox{.8}{\begin{tikzpicture}
  \begin{scope}[thin, gray]
  \draw (0, 0) rectangle (4.4, 1.5);
  \draw (0,0.5) to (4.4, 0.5);
  \draw (0, 1) to (4.4, 1);

  \draw [dashed] (4.8, 0.5) rectangle (9.2, 1.5);
  \draw [dashed] (4.8, 1) to (9.2, 1);

  \draw [dashed] (9.6, -0.5) rectangle (13.4, 0.5);
  \draw [dashed] (9.6, 0) to (13.4, 0);

  \draw [dashed] (4.8, -2) rectangle (9.2,-1);
  \draw [dashed] (4.8, -1.5) to (9.2,-1.5);

  \draw (0, -2.5) rectangle (4.4, -1);
  \draw (0, -2) to (4.4, -2);
  \draw (0, -1.5) to (4.4, -1.5);
\end{scope}

\node at (2.2, 1.25) {\scriptsize $q_0$, 2, 
  $\{c,d\}\{b\}\{c,d\}\{b\}$};
\node at (2.2, 0.75) {\scriptsize $q_0$, 0,
  $\{a\}\{b\}\{c,d\}\{b\}\{c\}$};
\node at (2.2, 0.25) {\scriptsize $q_0$, 2, 
$\{c,d\}\{b\}\{c,d\}\{b\}$};

\node at (7, 0.2) {\scriptsize \sf{synchronize}};
\node at (7, 1.25) {\scriptsize $q_0$, 2, $\{c,d\}\{b\}\{c,d\}\{b\}$};
\node at (7, 0.75) {\scriptsize $\blue{q_0}$, \blue{2}, $\{c,d\}\{b\}\{c,\blue d\}\{\blue b\}$};

\node at (11.5, 0.75) {\scriptsize \sf{expand}};
\node at (11.5, 0.25) {\scriptsize $q_0$, 2, $\{c,d\}\{b\}\{c,d\}\{b\}\{\blue{a}\}$};
\node at (11.5, -0.25) {\scriptsize $q_0$, 2, $\{c,d\}\{b\}\{c,d\}\{b\}\{\blue{a}\}$};

\node at (7, -0.75) {\scriptsize \sf{cut}};
\node at (7, -1.25) {\scriptsize $q_0$, 2, $\{c,d\}\{b\}\{c,d\}\{b\}\{a\}$};
\node at (7, -1.75) {\scriptsize $q_0$, 2, $\{c,d\}\{b\}\{c,d\}\{b\}\{a\}$};

\node at (2.2, -1.25) {\scriptsize $q_0$, 2, $\{c,d\}\{b\}\{c,d\}\{b\}\{a\}$};
\node at (2.2, -1.75) {\scriptsize $q_0$, 2, $\{c,d\}\{b\}\{c,d\}\{b\}\{a\}$};
\node at (2.2, -2.25) {\scriptsize $q_0$, 2, $\{c,d\}\{b\}\{c,d\}\{b\}$};

\begin{scope}[->, >=stealth]
  \draw (2, 0) to node [left] {\scriptsize $a$} (2, -1);
  \draw [gray, thin, dashed] (4.4, 1) to (4.8, 1);
  \draw [gray, thin, dashed] (9.2, 1) to (10, 0.5);
  \draw [gray, thin, dashed] (10, -0.5) to (9.2, -1);
  \draw [gray, thin, dashed] (5, -1.5) to (4.4, -1.5);
\end{scope}
\end{tikzpicture}}
\caption{Illustrating one transition in the AA that
  maintains an unbounded counter}\label{fig:unbouded-counter}
\end{figure}

A process $p$ thus only stores the steps of the FNF of its view $t_p$ starting with the one of index $\f(t_p, 2k-2)$. However, for later synchronisations, processes need to keep a way to align their views in order to do the stepwise union. We first choose to do it by using an \emph{unbounded counter}, remembering how many letters have been forgotten so far. The beginning of computation in Figure~\ref{fig:run-of-infinite-aa} remains identical, simply adding the initial state $q_0$, and a counter value $0$, before the view. 
As seen before, in the last configuration though, processes $p_1$ and $p_3$ can cut two letters still keeping the last steps having $2k-2=6$ letters. They thus update the state component (though it comes back to $q_0$ again), and increment the counter component by 2, and we reach the configuration:
\[\footnotesize
\begin{array}{c}p_1 \\ p_2 \\p_3\end{array}\begin{array}{|l|}\hline 
  q_0, 2, \{c,d\}\{b\}\{c,d\}\{b\}\\\hline 
  q_0, 0, \{a\}\{b\}\{c,d\}\{b\}\{c\}\\\hline 
  q_0, 2, \{c,d\}\{b\}\{c,d\}\{b\}\\\hline 
\end{array}\]
While reading an additional letter $a$, Figure~\ref{fig:unbouded-counter} presents the update performed by processes $p_1$ and $p_2$. They first synchronise their local views: process $p_2$ that had counter value $0$ first removes the first two letters of its view to obtain the same counter value $2$ as $p_1$, and then it aligns the rest of its view with $p_1$ and learns about the events $d$ and $b$. Then they expand their views by adding the step $\{a\}$. Even if the steps contain $7>2k-2$ letters in total, the first step, that contains two letters, cannot be cut, since the four last steps contain only $5<2k-2$ letters in total. 
If we read the letter~$b$ afterwards, processes $p_1$ and $p_3$ synchronise to the view of $p_1$ that had the most recent information, and the first step can be cut while incrementing the counter value twice (since the step that is removed contains 2 letters), to obtain the new configuration
\[\footnotesize
\begin{array}{c}p_1 \\ p_2 \\p_3\end{array}\begin{array}{|l|}\hline 
  q_3, 4, \{b\}\{c,d\}\{b\}\{a\}\{b\}\\\hline 
  q_0, 2, \{c,d\}\{b\}\{c,d\}\{b\}\{a\}\\\hline 
  q_3, 4, \{b\}\{c,d\}\{b\}\{a\}\{b\}\\\hline 
\end{array}\]
For each letter read, the corresponding processes first synchronise, then expand, and finally make the necessary cuts.

\smallskip

\noindent \textbf{Step 3.} So far, only the counter values make the set of local states infinite. We now propose to only remember these values modulo $2k$. The crux is Lemma~\ref{lem:pairOfProcsAtmostK-1Apart}: processes cannot lag behind one another by more than $k - 1$ letters. 
Because of this observation, the range of possible counter values of all processes, when added to the number of letters not yet forgotten, 
is included in an interval of values of the form $\{c, c+1, c+2, \ldots,
c+(k-1)\}$, i.e.,~$k$ values. When considered modulo $2k$, this becomes a cyclic
interval of $k$ values. In particular, there is a full range of the other $k$
values (modulo~$2k$) that cannot be taken by any process. We can thus
unambiguously find which of the processes is ahead of the other when
synchronising (even if its counter value seems lower than others, due to the
modulo counting). For instance, let us continue reading letters $dcb$ after
which we reach the following configuration still with counter values at most~$2k-1$:
\[\footnotesize
\begin{array}{c}p_1 \\ p_2 \\p_3\end{array}\begin{array}{|l|}\hline 
  q_3, 7, \{b\}\{a\}\{b\}\{c,d\}\{b\}\\\hline 
  q_0, 5, \{c,d\}\{b\}\{a\}\{b\}\{c\}\\\hline 
  q_3, 7, \{b\}\{a\}\{b\}\{c,d\}\{b\}\\\hline 
\end{array}\]
When reading the next letter $a$, the cutting of the first step would increase the counter value to $8$, which is rounded modulo $2k=8$ to $0$. We thus get: 
\[\footnotesize
\begin{array}{c}p_1 \\ p_2 \\p_3\end{array}\begin{array}{|l|}\hline 
  q_0, 0, \{a\}\{b\}\{c,d\}\{b\}\{a\}\\\hline 
  q_0, 0, \{a\}\{b\}\{c,d\}\{b\}\{a\}\\\hline 
  q_3, 7, \{b\}\{a\}\{b\}\{c,d\}\{b\}\\\hline 
\end{array}\] We arrive in the situation depicted on the top left of
Figure~\ref{fig:modulo-counter}. When reading the next letter $b$, processes
$p_1$ and $p_3$ have respective counter values $0$ and $7$. They add $6$ (the length of the suffix of the view which they explicitly remember) to get values $6$
and $5$ modulo $8$. As mentioned before, the list of counter values added to
the length of the remaining suffix, falls under a continuous range of at most
$k = 4$ values. These values should include $6$ and~$5$ --- hence $p_1$ is
indeed ahead of $p_3$ (despite having a smaller  counter value originally). If
it was the other way around, that is, $p_1$ is behind $p_3$, the range of
values would include $6, 7, 0, 1, 2, 3, 4, 5$, which has more than $k = 4$
values: this is not possible. Thus, $p_3$ can discard its first
step to synchronise. They expand the configuration by adding $b$ in a new
step of the FNF. They cut the first element, and update their
state to $q_3$ and counter value to $1$. This way, using modulo counting, the
processes are able to determine who is ahead and then perform the
synchronisation, expanding and cutting to get to the new local states that
maintain a suffix of their views.

\begin{figure}[tbp]
  \centering
  
  \scalebox{0.8}{\begin{tikzpicture}
    \begin{scope}[thin, gray]
    \draw (-0.2, 0) rectangle (4.2, 1.5);
    \draw (-0.2,0.5) to (4.2, 0.5);
    \draw (-0.2, 1) to (4.2, 1);
  
    \draw [dashed] (4.8, 0) rectangle (9.2, 0.5);
    \draw [dashed] (4.8, 1) rectangle (9.2, 1.5);
  
    \draw [dashed] (9.6, -0.25) rectangle (13.4, 0.25);
    \draw [dashed] (9.6, -1.25) rectangle (13.4, -0.75);
  
    \draw [dashed] (4.8, -2.5) rectangle (9.2,-2);
    \draw [dashed] (4.8, -1.5) rectangle (9.2,-1);
  
    \draw (-0.2, -2.5) rectangle (4.2, -1);
    \draw (-0.2, -2) to (4.2, -2);
    \draw (-0.2, -1.5) to (4.2, -1.5);
  \end{scope}
  
  \node at (2, 1.25) {\scriptsize $q_0$, 0, $\{a\}\{b\}\{c, d\}\{b\}\{a\}$};
  \node at (2, 0.75) {\scriptsize $q_0$, 0, $\{a\}\{b\}\{c, d\}\{b\}\{a\}$};
  \node at (2, 0.25) {\scriptsize $q_3$, 7, $\{b\}\{a\}\{b\}\{c, d\}\{b\}$};

  \node at (7, -0.2) {\scriptsize \sf{synchronise}};
  \node at (7, 1.25) {\scriptsize $q_0$, 0, $\{a\}\{b\}\{c, d\}\{b\}\{a\}$};
  \node at (7, 0.25) {\scriptsize $\blue{q_0}$, \blue{0}, $\{a\}\{b\}\{c, d\}\{b\}\{\blue a\}$};
  
  \node at (11.5, 0.5) {\scriptsize \sf{expand}};
  \node at (11.5, 0) {\scriptsize $q_0$, 0, $\{a\}\{b\}\{c, d\}\{b\}\{a\}\{\blue{b}\}$};
  \node at (11.5, -1) {\scriptsize $q_0$, 0, $\{a\}\{b\}\{c, d\}\{b\}\{a\}\{\blue{b}\}$};
  
  \node at (7, -0.75) {\scriptsize \sf{cut}};
  \node at (7, -1.25) {\scriptsize $\blue{q_3}$, \blue{1}, $\{b\}\{c, d\}\{b\}\{a\}\{b\}$};
  \node at (7, -2.25) {\scriptsize $\blue{q_3}$, \blue{1}, $\{b\}\{c, d\}\{b\}\{a\}\{b\}$};

  \node at (2, -1.25) {\scriptsize $q_3$, 1, $\{b\}\{c, d\}\{b\}\{a\}\{b\}$};
  \node at (2, -1.75) {\scriptsize $q_0$, 0, $\{a\}\{b\}\{c, d\}\{b\}\{a\}$};
  \node at (2, -2.25) {\scriptsize $q_3$, 1, $\{b\}\{c, d\}\{b\}\{a\}\{b\}$};
  
  \begin{scope}[->, >=stealth]
    \draw (2, 0) to node [left] {\scriptsize $b$} (2, -1);
    \draw [gray, thin, dashed] (4.2, 0.75) to (4.8, 0.75);
    \draw [gray, thin, dashed] (9.2, 0.75) to (10, 0.25);
    \draw [gray, thin, dashed] (10, -1.25) to (9.2, -1.75);
    \draw [gray, thin, dashed] (4.8, -1.75) to (4.2, -1.75);
  \end{scope}
  \end{tikzpicture}}
  \caption{Illustrating one transition in the finite AA $\B$ that
    counts modulo $2k$}\label{fig:modulo-counter}
  \end{figure}

\paragraph*{Computing the state in the DFA.} The final task is to be able to check whether a given global state of the built AA $\B$ is accepting or not. To do so, we map each global state of $\B$ to a unique state of the DFA $\A$ as follows.
First we synchronise all the processes: first $p_1$ and $p_2$ synchronise, then the joint state of $\{p_1, p_2\}$ is reconciled with the local state of $p_3$ and so on. During the synchronisation, the same care is taken on the modulo counter values. Once this is done, all the local states are the same. From the state $q$ stored in a local state, we compute the state $q'$ reached in $\A$ by reading all the letters of the common suffix, step by step, in any order. 
We declare the global state accepting in the AA $\B$ if and only if the state $q'$ of $\A$ is accepting.
  
The entire construction is carefully presented in~\Cref{app:full-proof-fairness} for the general case. The crucial point in the proof is to maintain that the state $q'$ computed above to decide the acceptance condition in $\B$ is indeed the state in $\A$ that would be reached if we had read the full trace.

The number of local states for each process is bounded by $n \times 2k \times |\Sigma|^M$ where $n$ is the number of states of $\A$, and $M$ is the maximum number of letters that should be maintained in the suffix of views. This value $M$ is bounded by $2k-2+(k-1)=3k-3$ since we may have to add up to $k-1$ letters to the first step in order to obtain a full step: indeed in a $k$-fair trace, every step of its FNF contains at most $k$ letters (otherwise there would be a letter $a$ independent of $k$ other letters in a factor $u$ of the trace, meaning that all processes reading $a$ would be starved during the factor $u$).

\newcommand{\pref}{\operatorname{pref}}

\paragraph*{Lower bound.} For the general case considering any trace-closed regular language as a specification,~\cite{GenGim10} presents a lower bound on the number of local states in any \emph{locally-rejecting} AA: a (deterministic) AA $\B$ is said to be locally-rejecting if for every process $p$, there is a set of reject states $R_p$ s.t. for every trace $t$, $\view_p(t) \notin \pref(\Lt(\B))$ iff the local state reached by $p$ on $t$ is in $R_p$, where $\pref(\Lt(\B))$ denotes the set of all ideals of all traces accepted by $\B$. If an AA synthesis construction maintains enough information so that the state reached by the global specification DFA $\A$ on $\view_p(t)$ can be deduced, then the synthesised AA is locally rejecting by a suitable choice of states $R_p$. It is remarked in~\cite{GenGim10} that every known Zielonka-type AA synthesis construction until then had induced a locally-rejecting automaton. After \cite{GenGim10}, the only AA synthesis procedure that we are aware of is \cite{AdsGas24}, but it is a different setting since the specification is written as a logical formula and the implementation is a cascade product of AAs. Notice that our algorithm also synthesises a locally-rejecting automaton since we can indeed derive the state reached by $\A$ on $\view_p(t)$. 
We are then able to show the following lower-bound, whose proof is presented in~\Cref{app:lower-bound}.

\begin{restatable}{theorem}{lowerBound}
There is a family of $n$-fair languages $L_n$ for which every locally-rejecting $AA$ has size at least $2^{\frac{n}{4}}$.
\end{restatable}

\section{Combining fairness with acyclic communication}\label{sec:combination}

In a fair DFA, every process participates in every cycle. In a DFA specification that is not fair, some set of processes can perform an unbounded number of actions, while the other processes wait. This could lead to processes having diverging views with an unbounded number of events in the difference. However, Krishna and Muscholl~\cite{KRISHNA2013109} have proposed an efficient way to synthesise an asynchronous automaton when the communication between processes happens in a hierarchical manner. Given a set $P$ of processes, and a distributed alphabet $(\Sigma, \loc)$, an undirected graph called the \emph{communication graph} $\gcomm$ is constructed as follows: vertices of $\gcomm$ are the processes in $P$; for $p_1, p_2 \in P$, there is an edge $(p_1, p_2)$ if $p_1$ and $p_2$ share a common action, i.e.~$\Sigma_{p_1} \cap \Sigma_{p_2}\neq \emptyset$. \cite{KRISHNA2013109} provides an efficient AA synthesis when the communication graph is a tree. In particular, this entails that every process synchronises only with its parent or its child in the communication graph. In their construction, the local states of a process $p$ are pairs of states $(\overleftarrow{q_p}, q_p)$ of the specification DFA $\A$ so that: 
 $\overleftarrow{q_p}$ is the state reached by $\A$ on reading $\jointview_p(t)$, the smallest ideal of $\view_p(t)$ containing all actions where $p$ synchronises with its parent in $\gcomm$; 
and $q_p$ is the state reached by the DFA $\A$ on reading $\view_p(t)$.
Their key technique is an elegant method to construct the state of $\A$ reached on a trace $t$ by looking at the local states $(\overleftarrow{q_p}, q_p)$ as above. We recall the details of the construction \Cref{sec:krishna-muscholl}.

In this section, we propose a \emph{tree-of-bags} architecture where essentially each node in the communication graph is replaced with a set (bag) of processes. When we impose a fairness condition on the DFA restricted to each of these bags, we are able to incorporate the construction of Section~\ref{sec:contrib-fair} into the method of~\cite{KRISHNA2013109}.

\smallskip
\noindent \textbf{Tree-of-bags architecture.} A distributed alphabet $(\Sigma, \loc)$ is said to form a \emph{tree-of-bags} architecture if the following conditions are all satisfied: 
\begin{itemize}
\item the set of processes $P$ can be partitioned into bags $\{ B_1, B_2, \dots, B_\ell\}$,
\item each bag $B_j$ has a special process $o_j$ called its \emph{outer process}  --- the set of all outer processes is denoted as $O$, and the rest of the processes are called \emph{inner processes};
\item the communication graph restricted to $O$ forms a tree --- hence, one outer process can be designated as the root, and for every other outer process $o_j$, there is a unique outer process which is the parent in this tree, which we denote as $\parent(o_j)$;
\item every inner process communicates within its bag: for a process $\iota \in B_j \setminus \{o_j\}$, there is no edge outside $B_j$ in the communication graph.
\end{itemize}

\begin{wrapfigure}{r}{4cm}
\vspace{-9mm}
\scalebox{0.8}{\begin{tikzpicture}
  \begin{scope}[xshift=11cm, yshift=1cm, scale=0.8]
    \begin{scope}[every node/.style={circle, fill, inner sep=2pt}]
      \node (0) at (0,0) {};
      \node (1) at (-1,-1) {};
      \node (2) at (1,-1) {};
      \node (3) at (-1.5, -2) {};
      \node (4) at (-0.5, -2) {};
      \node (5) at (1.5, -2) {};
    \end{scope}

    \begin{scope}[every node/.style={rectangle, fill, blue, inner sep=2pt}]
      \node (a) at (-0.4, 0.5) {};
      \node (b) at (0.4, 0.5) {};
      \node (c) at (-1.4, -0.3) {};
      \node (d) at (-1.4, -0.8) {};
      \node (e) at (-1.7, -0.5) {};
      \node (f) at (1.2, -0.6) {};
      \node (g) at (1.5, -0.3) {};
      \node (h) at (1.8, -0.6) {};
      \node (i) at (1.5, -0.8) {};
      \node (j) at (-2.2, -1.7) {};
      \node (k) at (-2.5, -2) {};
      \node (l) at (-2.1, -2.2) {};
      \node (m) at (2, -1.8) {};
      \node (n) at (2, -2.2) {};
    \end{scope}

    \begin{scope}[thin, gray]
      \draw (0) to (a) to (b) to (0);
      \draw (0) to (1);
      \draw (0) to (2);
      \draw (1) to (d) to (c) to (e) to (d);
      \draw (1) to (3);
      \draw (1) to (4);
      \draw (2) to (5);
      \draw (2) to (f);
      \draw (f) to (g);
      \draw (f) to (i);
      \draw (i) to (g);
      \draw (i) to (h);
      \draw (3) to (j) to (k) to (l) to (j);
      \draw (3) to (k);
      \draw (5) to (m) to (n) to (5);
    \end{scope}

    \begin{scope}[very thin, red, dashed]
      \draw (0) .. controls (-2, 1) and (2,1) .. (0);
      \draw (1) .. controls (-3, -1) and (-1, 1) .. (1);
      \draw (2) .. controls (1, 1) and (3, -1) .. (2);
      \draw (3) .. controls (-3, -0.5) and (-3, -3.5) .. (3);
      \draw (5) .. controls (2.5, -0.5) and (2.5,-3.5) .. (5); 
    \end{scope}
  \end{scope}
\end{tikzpicture}}
\vspace{-1.5cm}
\end{wrapfigure}
On the right, we depict an example where the red dashed bubbles represent the bags, each bubble has an outer process shown as a black circle, and the blue squares are the inner processes. We write $(\Sigma, \loc, \BB)$ to denote a tree-of-bags architecture as above, with $\BB = \{B_1, B_2, \dots, B_\ell\}$. For a bag $B \in \BB$, we write $o(B)$ for the outer process of $B$. Also, we let $\Sigma^{\inn}(B) = \{ a \in \Sigma \mid \loc(a) \incl B\}$ be the set of actions for which only the processes in $B$ participate. We will sometimes refer to $\Sigma^{\inn}(B)$ as the bag alphabet of $B$.

\begin{definition}
Let $(\Sigma, \loc)$ be a distributed alphabet forming a tree-of-bags architecture $(\Sigma, \loc, \BB)$. A DFA specification $\A$ over $(\Sigma, \loc)$ is said to be \emph{fair for the tree-of-bags architecture} $(\Sigma, \loc, \BB)$ if in every loop of $\A$, either all processes in a bag participate, or none of them participates.
\end{definition}

As for the notion of fairness studied in Section~\ref{sec:fair}, one can test if the DFA $\A$ is fair for a tree-of-bags architecture $(\Sigma, \loc, \BB)$. For each bag $B \in \BB$ replace the transitions labelled with $\Sigma \setminus \Sigma^{\inn}(B)$ by $\varepsilon$-transitions, and find if there exists a natural number~$k_B$ such that the resulting automaton is $k_B$-fair. Such constants must exist for each bag $B \in \BB$ for $\A$ to be fair for $(\Sigma, \loc, \BB)$.
In other words, for every trace in $L(\A)$, the trace obtained by only keeping events labelled by $\Sigma^{\inn}(B)$ is $k_B$-fair.
It is also possible to compute in polynomial
time such constants $k_B$ for all bags (by using the same algorithm described in~\Cref{app:deciding-kfair} for computing the fairness parameter for a DFA). In our later explanations, we suppose
that such constants $k_B$ are known, and we use them to give a complexity bound
of our construction.

\smallskip
\noindent \textbf{The construction.} Here is the main challenge: on reading a trace $t$, we need for each bag $B$, the state reached by $\A$ on $\view_B(t)$. Let us write $\delta(q_0, \view_B(t))$ for this state. Once we know it, we can use the $\state$ function of \cite{KRISHNA2013109} to compute the state reached by $\A$ on the whole trace $t$. So let us now focus on how we can compute $\delta(q_0, \view_B(t))$. Firstly, recall that $\view_B(t) = \bigcup_{\iota \in B} \view_{\iota}(t)$. If we maintain the individual views $\view_\iota(t)$ for each $\iota \in B$, in FNF, we could take the stepwise union to get $\view_B(t)$. Since we assumed $\A$ is fair when restricted to each bag, we are tempted to apply the construction of Section~\ref{sec:contrib-fair} and maintain only a finite suffix of these views. However, there are new challenges to overcome.
Since the outer process $o(B)$ communicates with processes outside bag $B$,  $\view_o(t)$ includes letters outside the bag alphabet. Further, as other inner processes $\iota \in B$ can potentially synchronise with $o$, $\view_{\iota}(t)$ may have also letters from outside the bag alphabet $\Sigma_B$. Hence we cannot directly cut the prefixes of the views based on the parameter $k_B$. 

Here are the key ideas. Firstly, we consider $\trest$, the trace $t$ restricted to the bag alphabet and maintain $\view_\iota(\trest)$ in FNF. Each process $\iota \in B$ maintains $\view_{\iota} (\trest)$. With this information, we can compute $\view_B(\trest)$. Secondly, we observe that $\view_B(t)$ can be written with $\view_o(t)$, $\view_B(\trest)$ and $\view_o(\trest)$. 

\begin{restatable}{lemma}{viewB}
 If $t$ is a trace, $B \in \BB$, and $o = o(B)$, we have $\view_B(t) = \view_o(t) \cdot t'$ where $t'$ is the trace obtained by removing $\view_o(\trest)$ from $\view_B(\trest)$. 
\end{restatable}

\begin{algorithm}[t]
  \caption{Computing the state $\delta(q_0, \view_B(t))$}\label{alg:bagview}
\begin{algorithmic}[1]
  \Require A bag $B$, a trace $\view_\iota(\trest)$ for all $\iota \in B$, $q_o = \delta(q_0, \view_o(t))$ where $o = o(B)$ 
  \Ensure $q = \delta(q_0, \view_B(t))$ 
  \State $\hat{t} \gets$ pointwise union of $\Foata(\view_\iota(\trest))$ over all $\iota \in B$ 
  \State $e \gets$ maximal element of $o$ in $\hat{t}$
  \State $t' \gets$ trace obtained by removing $\ideal{e}$ from $\hat{t}$
  \State $q \gets \delta(q_o, t')$
\end{algorithmic}
\end{algorithm} 

Algorithm~\ref{alg:bagview} shows how to compute $\delta(q_0, \view_B(t))$ using $\delta(q_0, \view_o(t))$ and $\view_\iota(\trest)$ for all $\iota \in B$. It identifies the trace $t'$ in Line 3, and then does the required computation in Line 4. The last task is to be able to implement Algorithm~\ref{alg:bagview} using finite information as states. This is where the fairness assumption on each bag helps, allowing us to maintain only a bounded suffix of $\view_\iota(t)$ for each $\iota$. For outer processes we also maintain a pair $(\overleftarrow{q_o}, q_o)$ where $q_o= \delta(q_0, \view_o(t))$ is the state reached by $\A$ on reading $\view_o(t)$ and $\overleftarrow{q_o} = \delta(q_0, \jointview_o(t))$, to mimick the construction of \cite{KRISHNA2013109} for the outer processes. All the details of the construction are presented in Appendix~\ref{app:full-construction-treeofbags}.

\begin{theorem}
  For a DFA $\A$ that is fair for a tree-of-bags architecture $(\Sigma, \loc, \BB)$, there exists a finite AA $\B$ over $(\Sigma, \loc)$ such that $\Lt(\A)=\Lt(\B)$, where the set of local states for outer processes is of size bounded by $O(n^3 \times k\times |\Sigma|^{3k-3})$, and that of inner processes is bounded by $O(k\times |\Sigma|^{3k-3})$, where $n$ is the number of states of $\A$, and $k$ is the maximum of $k_B$ over all bags $B$.
\end{theorem}


\section{Implementation and experimentation}
\label{sec:implementation}

We have implemented the synthesis algorithm of
Section~\ref{sec:contrib-fair}. It is available both as a docker container \cite{10.5281/zenodo.18165965}, and as a git repository \url{https://gitlab.lis-lab.fr/benjamin.monmege/faast}. 
The tool is implemented in Python. Each Python class contains tests that
clarify the usage of the tool, and a Jupyter notebook is provided that shows how to use the code, and some of the results we are able to obtain from it. There are no dependencies other than Python 3.

%
Thanks to the tool, our algorithm was applied to the classic \emph{Dining Philosophers} problem~\cite{Dij71}, a fundamental case study for concurrency and fairness issues that we have recalled in the introduction.
The DFA modelling every possible sequence of actions of the $n\geq 2$ philosophers is trace closed but not fair (unless there are $n=2$ philosophers). One possible goal is to synthesise a distributed implementation of the dining philosophers that guarantees liveness (absence of starvation of philosophers, here meant in the sense that every philosopher should be able to do an action often enough). 

However the DFA is not fair. We can still adapt the algorithm of Section~\ref{sec:contrib-fair} to allow for an unfair DFA in input, as well as a fairness parameter $k\in\N$. The modified algorithm produces an AA that accepts all $k$-fair traces in the language of the trace-closed DFA (possibly none). The construction is modified to only generate a state if its FNF part is $k$-fair (i.e.~all infixes of length $k$ of all possible linearisations of the FNF involve all processes). We are then able to obtain an AA for the dining philosophers by using as fairness parameter $2n$. The value $2n$ is chosen to ensure that at least one complete set of actions for any philosopher occurs over any sufficiently long execution, thereby ensuring fairness.\footnote{If, e.g., the DFA forbids a philosopher to put back a chopstick as long as they did not eat, the built AA ensures that every philosopher eats often.}

Notice that the tool proposes to apply a small optimisation, modifying the \emph{cut} step: instead of keeping at least $2k-2$ letters in the FNF, we cut as many steps as possible in the FNF, by removing all steps containing letters that are all known by every process (i.e.~that are fully in the views of all processes). This is a correct optimisation, and we depict the various results we obtain with the tool on a small example in Appendix~\ref{app:tool}.

As an alternative solution for the dining philosophers, we could apply directly the construction of Section~\ref{sec:contrib-fair} to the intersection of the DFA with the set of all $2n$-fair traces. Notice that the set of all $2n$-fair traces is indeed a regular trace language that is fair, but the DFA that recognises it is huge (at least exponential in $n$). Thus, the intersection results in a huge DFA that is long to produce, and thus impracticable for the synthesis algorithm afterwards.

\paragraph{On-the-Fly Generation Optimisation.}
Our implementation also introduces a crucial optimisation to manage the potential state-space explosion of the AA, whose size is exponential with respect to the fairness parameter $k$. It consists of an  on-the-fly generation of the AA, that computes and stores the states and transitions of the AA only when they are effectively needed to compute an actual run from the initial global state. Instead of computing the transition functions $\delta_a$ for all letters $a$, we only write the recipe to generate all possible transitions. Afterwards, the collections of states and transitions are populated dynamically only: when we run the AA (to check whether a given word is accepted, or via a random exploration of the paths in the global semantics), it is only at this moment that the new local states, global transitions, and global accepting states are computed and stored.
This approach also ensures that the constructed AA is limited to its reachable portion from the initial state, which is essential in practice for analyzing complex distributed systems where the non-reachable global state-space would be too large to compute.

\section{Conclusion} 

We have studied Zielonka's theorem to synthesise AA from DFA specifications, in the context of fair specifications. 
We have strengthened this result to restrict fairness to a subset of processes with an outer process, where outer processes have no cyclic dependencies. In these special cases, we obtain AA whose sets of local states do not depend on the number of processes, contrary to all previous methods where an exponential dependency in this number of processes is unavoidable. Furthermore, our construction is, arguably, simpler to understand conceptually. It would be interesting to study how the fairness restriction could simplify the construction of the gossip automaton, and thus the other proofs of Zielonka's theorem. As future works, we could also consider infinite traces, where the fairness condition makes even more sense. Similarly, we could investigate the fairness condition in the context of a specification that is stated as a logical formula, for instance LTL or PDL, hoping that our techniques can lead to beneficial results in the context of translating logical formulae into AA. We would also like to study extensions of our techniques to the setting of $k$-connectedly communicating processes~\cite{MadThi05}, where the fairness restriction can be relaxed for processes that are no longer communicating with each other in the future.

\bibliographystyle{splncs04}
\bibliography{fair-Zielonka-etaps26.bib}

\appendix

\section{Appendix for Section~\ref{sec:prelim}}

\prefixFoata*
\begin{proof}
  We proceed by induction on $\Foatalength s$, for a fixed trace $t$. 
  
  If $\Foatalength s = 1$, then $\Foata(s)$ is a single step containing the (independent) letters that label all the events of $s$. All such events are minimal events of $s$, and thus of $t$, since $s$ is an ideal of $t$. Therefore, the first step of $\Foata(t)$ contains all these letters too. 

  Suppose now that $\Foata(s)=S_1 \dots S_{k-1} S_k$ with $k\geq 2$. We also let $\Foata(t)=T_1\dots T_m$. By induction hypothesis, we have that for all $i \leq k-1$, $S_i \subseteq T_i$. Let $a \in S_k$. Since $s$ is an ideal of $t$, $a$ appears in $T_\ell$ for some $\ell \leq m$. We show that $\ell=k$. Since $a\in S_k$, we know that there exists $b \in S_{k-1}$ that is dependent on $a$, such that the corresponding events of $s$ (and thus of $t$) satisfy $e \lessdot f$. Since $b \in S_{k-1}\subseteq T_{k-1}$, we have $\ell \geq k$. Similarly, since $a\in T_\ell$, there exists $c \in T_{\ell-1}$ that is dependent on $a$, such that the corresponding events of $t$ satisfy $e' \lessdot f$. Since $s$ is an ideal of $t$, $e'$ is also an event of $s$, and thus $c$ appears in a step $S_{\ell'}$ with $\ell'<k$. 
  This means that $c \in S_{\ell'}\incl T_{\ell'}$ which implies that $\ell = \ell'+1 \leq k$. We thus conclude that $\ell = k$, so that $a\in T_k$. Since this holds for all letters $a\in S_k$, we have $S_k\subseteq T_k$. 
\end{proof}

\viewFoata*
\begin{proof}
  For all $p \in X$, $\view_p(t)$ is an ideal of $\view_X(t)$. From Lemma~\ref{lem:prefixFoata}, we thus know, for all $i\leq \Foatalength{\view_p(t)}$, 
  $\Foata(\view_p(t))_i \subseteq \Foata(\view_X(t))_i$. This inequality is even true for $\Foatalength{\view_p(t)}<i\leq \Foatalength{\view_X(t)}$, since we have added empty sets as steps at the end of FNF. Thus, we have $\bigcup\limits_{p\in X} \Foata(\view_p(t))_i \subseteq \Foata(\view_X(t))_i$.
  
  Reciprocally, if some event labelled by letter $a$ occurs in $\view_X(t)$ then it appears in some $\view_p(t)$ for $p \in X$, by definition of the views. Thus, we get the desired equality.
\end{proof}

\section{Proof of Section~\ref{sec:fair}}

\pairOfProcsAtmostKApart*
\begin{proof}
    Let $p\in P$. Let $s$ the factor (even the suffix) of $t$ after $\view_p(t)$: i.e., we have $t = \view_p(t) s$. Since $\view_p(t)$ is the view of process $p$, no event of $s$ has a label in $\Sigma_p$. Since $t$ is $k$-fair, the length of $s$ is bounded by $k-1$, and thus $\length{t}\leq \length{\view_p(t)}+k-1$, so that $\length{\view_p(t)}\leq \length{t}< \length{\view_p(t)}+k$. 

    Let $p, p' \in P$. By the previous bounds, we have $0\leq \length{t}-\length{\view_p(t)}< k$, and $0\leq \length{t}- \length{\view_{p'}(t)}<k$. Thus, $-k < \length{\view_p(t)}- \length{\view_{p'}(t)} < k$ which proves the lemma.
\end{proof}

\completeTraceBeyondKFoataComponents*
\begin{proof}
  Let $t_p = \view_p(t)$, and $j=\f(t_p, k-1)$. By using Lemma~\ref{lem:prefixFoata}, we know that 
  for all $i < j$, $\Foata(t_p)_i \subseteq \Foata(t)_i$. If this inequality is an equality, for all $i$, the lemma holds. Otherwise, let 
  $\alpha$ be such that $\Foata(t_p)_\alpha \subsetneq \Foata(t)_\alpha$, and for all $i < \alpha$, $\Foata(t_p)_i = \Foata(t)_i$. Let $a \in \Foata(t)_\alpha \setminus \Foata(t_p)_\alpha$ and $p' \in \loc(a)$. For all $i \geq \alpha$, $\Sigma_{p'} \cap \Foata(t_p)_i = \emptyset$ since otherwise the event labelled by $a$ in the step $\Foata(t)_\alpha$ would also occur in $t_p$.
  Let $u$ be a linearisation of a trace whose FNF is $\Foata(t_p)_\alpha \cdots \Foata(t_p)_{\Foatalength{t_p}}$. Since $u$ is a factor of a linearisation of a $k$-fair trace, and $p'\notin\loc(u)$, the length of $u$ is less than $k$. Thus, $j\leq \alpha$. In particular, for all $i \leq j$, we have $i\leq \alpha$ and thus $\Foata(t)_i = \Foata(t_p)_i$ as expected.
\end{proof}

\everyoneKnowsBeyondKFoataComponents*
\begin{proof}
  Let $t_p = \view_p(t)$, $j=\f(t_p, 2k-2)$. By application of
  Lemma~\ref{lem:completeTraceBeyondKFoataComponents}, we already know that for
  all $i< \f(t_p, k-1)$ (and thus $i < j$ since $\f(t_p, k-1) \geq
  j$), $\Foata(t)_i = \Foata(t_p)_i$.

  Let $p'\in P$. 
  By Lemma~\ref{lem:prefixFoata}, we know that this implies that  
  $\Foata(\view_{p'}(t))_i \subseteq \Foata(t)_i = \Foata(t_p)_i$, where the last equality comes from the above result. In particular,
  $\Foata(\view_{p'}(t_p))_i \subseteq \Foata(t_p)_i$ (since all events of $\view_{p'}(t_p)$ are in $\view_{p'}(t)$). The trace $t_p$, that has length at least $2k-2$, can be written as $\view_{p'}(t_p) t'$.
  Since $t_p$ is $k$-fair (as an ideal of a $k$-fair trace), and $p'$ does not
  participate in any letter of $t'$, a factor of $t_p$, the length of $t'$ is less than $k$. Thus, the length of $\view_{p'}(t_p)$ is at least $2k-2-|t'| \geq k-1$. In
  particular, $f(\view_{p'}(t_p), k-1) \geq j$. By
  Lemma~\ref{lem:completeTraceBeyondKFoataComponents} applied on $t_p$ and
  process $p'$, we know that for all $\ell<f(\view_{p'}(t_p), k-1)$ (and thus for all $\ell<j$),
  $\Foata(t_p)_\ell = \Foata(\view_{p'}(t_p))$.
\end{proof}

\fairDFALanguage*
\begin{proof}
  If $\A$ is $k$-fair, let $t \in \Lt(\A)$, and $u$ a factor of length at least $k$ of a linearisation $v_1 u v_2$ of $t$. Let $q = \delta(q_0, v_1)$, and $q' = \delta(q, u)$. Since $\delta(q', v_2) \in F$, $q'$ is co-reachable. Since $\A$ is $k$-fair, $\loc(u) = P$.

  Reciprocally, if $\Lt(\A)$ is $k$-fair, assume that $\A$ is not $k$-fair, i.e.~there exist linearisations $u, v, w$ of traces such that $|v| \ge k$ and $\delta(q_0, uvw) \in F$, but $\loc(v) \neq P$. Since $uvw \in \Lt(\A)$ with $|v| \ge k$ and $\loc(v) \neq P$, the $k$-fair property of $\Lt(\A)$ is contradicted. Thus $\A$ is $k$-fair.
\end{proof}

\subsection{\texorpdfstring{Deciding $k$-fairness of a DFA}{Deciding k-fairness of a DFA}}\label{app:deciding-kfair}

\begin{algorithm}[H]
  \caption{Check Fairness of an DFA}\label{alg:k-fair-check}
  \begin{algorithmic}[1]
    \Require $\A = (Q, \Sigma, q_0, \Delta, Q_f)$ is a trim diamond DFA.
    \For {$k \gets \{1, \dots, |Q|\}$}
      \For {$p \gets P$}
      \State $\mathit{fair}(k, p) \gets \mathit{False}$
        \State $G_p \gets \{(q, q') \in Q \times Q \mid \exists a \in \Sigma \setminus \Sigma_p.\ q' \in \Delta(q, a)\}$
        \If {$(G_p)^k = \emptyset$} \Comment{$(G_p)^k$ is the relation $G_p$ composed with itself $k$ times}
          \State $\mathit{fair}(k, p) \gets True$
        \EndIf
      \EndFor
      \If {$\forall p \in P,\ \mathit{fair}(k, p) = \mathit{True}$}
        \Return $k$
      \EndIf
    \EndFor
    \State \Return $\bot$
  \end{algorithmic}
 \end{algorithm}
 
 \begin{lemma}
   For a trim DFA $\A$ satisfying the diamond property, either \Cref{alg:k-fair-check} returns the least positive integer $k$ such that $\A$ is $k$-fair, or \Cref{alg:k-fair-check} returns $\bot$ if $\A$ is not $k$-fair for all $k$. \Cref{alg:k-fair-check} runs in time $O(|P| \cdot |Q|\cdot (|\Delta| + |Q|^4))$
   \label{lem:alg:k-fair-check-correct}
 \end{lemma}
 
 \begin{proof}
  To begin with: suppose $\A$ is not $|Q|$-fair. Then, there is a word $u = a_1 a_2 \dots a_{|Q|}$, such that $\loc(u) \neq P$ and $\delta(q, u)$ is co-accessible. Let the run of $\A$ on $u$ be: $q_0 \xrightarrow{a_1} q_1 \cdots \xrightarrow{a_{|Q|}} q_{|Q| + 1}$. Since there number of states is $\le |Q|$, some state repeats in this run. This shows that there is a loop in the DFA where not all processes participate. Hence $\A$ is not fair. This shows that $\A$ is fair iff $\A$ is $|Q|$-fair. Now, let us show the correctness and complexity of the algorithm.

   If $\A$ is $k$-fair (for $k \le |Q|$) and not $k'$-fair for all $k' < k$, then for each $k' < k$, there exists a pair of states $(q, q') \in Q \times Q$, a word $u$ of length $k'$, and a process $p \in P$ such that $q' \in \Delta(q, u)$ and $p \notin \loc(u)$. Thus, $(q, q') \in (G_p)^{k'}$ and $\mathit{fair}(k', p) = \mathit{False}$ and the condition on line 7 evaluates to $\mathit{False}$. At iteration $k$, $(G_p)^k = \emptyset$ for all $p \in P$ for otherwise if $(q, q') \in (G_p)^k$ for some $p \in P$ then there exists a word $u$ of length $k$ such that $q' \in \Delta(q, u)$ and $p \notin \loc(u)$ which is not possible. Therefore, the algorithm returns the least non-negative integer $k$ such that $\A$ is $k$-fair.
 
   If $\A$ is not fair, then it is not $|Q|$-fair, and thus for all $k \le |Q|$ there exists $p \in P$ such that $\mathit{fair}(k, p) = \mathit{False}$. Hence, the algorithm returns $\bot$.

 There are at most $|Q| \times |P|$ iterations of lines 3--6. The graph $G_p$
 in Line 4 can be computed in $O(|\Delta|)$ time and $(G_p)^k$ in line 5 can be
 computed in $O(k \cdot |Q|^3)$ time. Line~7 takes $O(|P|)$ time. Thus, the
 algorithm runs in time $O(|P|\cdot |Q| \cdot (|\Delta| + k |Q|^3))$. Since $k\in O(|Q|)$, we deduce the announced complexity. 
 \end{proof}

\section{Full proof of Theorem~\ref{thm:fair-AA}}\label{app:full-proof-fairness}

We first build an infinite AA from the DFA specification based on the FNF, and then simplify it with the counter to make it finite in the end.

\subsection{The construction of an infinite AA based on the Foata normal form}\label{sec:infinite}

Towards the proof of Theorem~\ref{thm:fair-AA}, we first fix a DFA $\A = (Q, \Sigma, q_0, \Delta, F)$ satisfying the diamond property over a concurrent alphabet $(\Sigma, I)$, and a distributed alphabet $(\Sigma, \loc)$ such that $I_\loc = I$, and we show that there exists an infinite AA over $(\Sigma, \loc)$ that recognises it. This is of course less powerful than Zielonka's theorem, but the construction we use (that consists in an unfolding based on the FNF) is a prerequisite for the later simplification, in case the DFA 
$\A$ is fair. In this section, we thus do not rely on any fairness property. 

The infinite AA that we build is $\B^\infty = ((Q^\infty_p)_{p \in P}, \Sigma, q^\infty_0, (\delta^\infty_a)_{a \in \Sigma}, F^\infty)$. For each process $p\in P$, we let $Q^\infty_p$ be the set of FNF of possible views of this process along an execution: $Q^\infty_p = \{\varepsilon\} \uplus \{\Foata(\view_p(t))\mid t \text{ non-empty trace}\}$, where $\varepsilon$ denotes the FNF of the empty trace that we use at the beginning of the execution.
We thus let $q^\infty_0 = (\{\varepsilon\}, \ldots, \{\varepsilon\})$ since no process has viewed any letter of the trace. We now introduce two operations in order to define the transitions and final states.

First, we explain the operation $\synchronise$ that consists in merging the views of a subset of processes (either all processes at the end of the trace, or just the processes that will read a common letter of the trace). For a subset $X\subseteq P$ of processes, and local states $(\varphi_{p})_{p\in X}$ for all these processes, we let $\synchronise((\varphi_{p})_{p\in X}) = \bigcup_{p\in X} \varphi_{p}$ be the stepwise union of the FNF.
By Lemma~\ref{lem:viewFoata}, this implies that: 

\begin{lemma}\label{lem:infiniteB-synchronisation}
  For all traces $t$, 
  we have $\synchronise((\Foata(\view_p(t)))_{p\in X}) = \Foata(\view_X(t))$. 
\end{lemma}

For instance, before reading the second $b$ in Figure~\ref{fig:run-of-infinite-aa}, processes $p_1$ and $p_3$ synchronise their respective FNF of views, by adding $c$ and $d$ in their last steps respectively. Before reading the first  $c$, processes $p_2$ and $p_3$ synchronise: $p_2$ learns about the new step $\{b\}$ in the FNF. 

After the synchronisation step, we explain the operation $\expand$ that consists in adding a new letter to the trace. For a subset $X\subseteq P$ of processes, a sequence $S_1\cdots S_m$ of steps (the local state common to all these processes since they first synchronised), and a letter $a$ such that $\loc(a) = X$, we let $\expand(S_1\cdots S_m, a) = S_1\cdots S_m\{a\}$. This expansion builds again a FNF, since there exists a letter in the step $S_m$ where one of the processes of $X$ participates (otherwise this step would simply not appear in the synchronisation step), and thus that is dependent on $a$.

With the help of these two functions, we define the transition functions. For $a\in \Sigma$, letting $X= \loc(a)$, and local states $(\varphi_{p})_{p\in X}$, we let $S'_1\cdots S'_{m'} = \expand(\synchronise((\varphi_{p})_{p\in X}), a)$, and define $\delta^\infty_a((\varphi_{p})_{p\in X}) = (S'_1\cdots S'_{m'}, \ldots, S'_1\cdots S'_{m'})$. By induction, we maintain the following invariant over the state reached in $\B^\infty$ after having read a certain trace: 
\begin{lemma}\label{lem:Binfty}
For all traces $t=[w]$ with $w\in \Sigma^*$ and processes $p\in P$, the local state of $p$ reached in $\B^\infty$ after having read $w$ is $\Foata(\view_p(t))$.
\end{lemma}

We can also define the accepting global states of $\B^\infty$. For a global state $(\varphi_{p})_{p\in P}$, we put it in $F^\infty$ if and only if the state $q'$ reached in $\A$ after reading any word $w$ such that $\Foata([w]) = \synchronise((\varphi_{p})_{p\in P})$ is accepting in $\A$: since $\A$ satisfies the diamond property, this definition does not depend on the choice of $w$ in the equivalence class $[w]$. Once again using Lemma~\ref{lem:viewFoata}, and the previous invariant, we deduce that after having read a trace $t$ in $\B^\infty$, we reach the global state $(\varphi_{p})_{p\in P}$ such that $\synchronise((\varphi_{p})_{p\in P})$ is exactly $\Foata(t)$: we thus accept $t$ if and only if any of its linearisations $w$ is accepted in $\A$. This ends the proof that 

\begin{proposition}\label{prop:infinite-AA}
  The AA $\B^\infty$ recognise the language of the DFA $\A$. 
\end{proposition}

Towards the result of Theorem~\ref{thm:fair-AA}, we modify the infinite AA $\B^\infty$ built in the Proposition~\ref{prop:infinite-AA}. We first introduce counters in order to only keep suffixes of the views for each process. The counters will be unbounded, and thus the AA would still be infinite. We then explain how to make the counters finite, by modulo counting. We now heavily relies on the hypothesis that $\A$ is supposed to be $k$-fair (without loss of generality by Proposition~\ref{prop:fairDFALanguage}). 

\subsection{Introduction of unbounded counters to cut the views of processes}\label{sec:unbounded-counter}

We build upon Lemma~\ref{lem:everyoneKnowsBeyond2KFoataComponents}, stating that if the view $t_p$ of a process $p$ is of length at least $2k-2$, then all other processes know entirely all the steps of the FNF of the trace that has been read so far, except possibly the steps starting from the one of index $f(t_p, 2k-2)$. We thus build another infinite AA $\B^\N= ((Q^\N_p)_{p \in P}, \Sigma, q^\N_0, (\delta^\N_a)_{a \in \Sigma}, F^\N)$ 
that recognises the same language. The exponent $\N$ is there to remember that the only infinite part in this new state space is a counter taking values in $\N$. 

We thus need to introduce a new operation $\cut$ that consists in only keeping the shortest suffix of the FNF that contains at least $2k-2$ letters, also keeping a counter (a natural number) to remember how many letters have been forgotten so far, and the state of $\A$ that was reached after having read those letters (this is crucial so that processes can still collectively compute the state reached in $\A$ at the end of the trace).

The set of states $Q^\N_p$ of a process $p\in P$ contains all triples $(q, c, \varphi)$ with $q\in Q$, $c\in \N$, and $\varphi$ a FNF of the form $S_1 S_2 \ldots S_m$ such that the number of letters of $S_2\ldots S_m$ is at most $2k-2$
(optionally, we may also restrict ourselves to FNF of traces that can be the suffix of a view of process $p$: in particular, the last step, if it exists, must contain a letter in which $p$ participates).

Starting from a tuple $(q, c, \varphi)$ with $q\in Q$, $c\in \N$, $\varphi = S_1 \cdots S_m$ a FNF, we define $\cut(q, c, \varphi)$. If 
the number of letters in $S_2\cdots S_m$ is less than $2k-2$, we let $\cut(q, c, \varphi)=(q, c, \varphi)$. 
Otherwise, the tuple is not a local state since the Foata component is too long. We let $\cut(q, c, \varphi)=
(q', c+\sum_{\ell=1}^{j-1}|S_\ell|, S_{j}\cdots S_{m})$ where: 
\begin{itemize}
\item $j$ is the largest index such that $S_{j}\cdots S_{m}$ contains at least $2k-2$ letters;
\item $q'$ is the unique state of $\A$ reached from $q$ by reading all letters of $S_1$, and then all letters of $S_2$, etc.~until all letters of $S_{j-1}$ (for each step, letters can be read in any order, since $\A$ satisfies the diamond property and all letters of a step are independent).
\end{itemize}

We also redefine the functions $\synchronise$ and $\expand$ over these new states. Let $a$ be a letter, with $X=\loc(a)$, and $(q_p, c_p, \varphi_{p})$ be the local state of process $p$, for all $p\in X$. We let $\synchronise((q_p, c_p, \varphi_{p})_{p\in X}) = (q', c', \varphi')$ defined as follows: 
\begin{itemize}
\item we consider any process $p'$ with $c_{p'}=\max(c_p)$, i.e.~one of the processes that has the most recent information, counterwise;
\item we let $q'=q_{p'}$ and $c'=c_{p'}$, trusting the information of process $p'$;
\item for all processes $p\in X$, if $\varphi_{p} = S_1\cdots S_m$, we let $j$ be such that $S_1\ldots S_{j}$ contains $c_{p'}-c_p$ letters (all those letters are known by every process already, since $p'$ has at least $2k-2$ letters after the $c_{p'}$-th letter, and we are thus ensured that there is a union of steps that contain $c_{p'}-c_p$ letters), and we let $\psi_{p} = S_{j+1} \cdots S_m$: we thus simply forget every step that is too old (and known by everyone) with respect to the newest information;
\item finally, we let $\varphi' = \bigcup_{p\in X} \psi_{p}$, thus generalising the operation $\synchronise$ defined in $\B^\infty$. 
\end{itemize}
We also let $\expand((q', c', \varphi'), a) = (q', c', \varphi'\{a\})$, generalising the previous definition of $\expand$. 

To define $B^\N$, we thus consider as initial state $q^\N_0 = ((q_0, 0, \{\varepsilon\}), \ldots, (q_0, 0, \{\varepsilon\}))$; for all letters $a\in \Sigma$, and local states $(q_p, c_p, \varphi_{p})_{p\in X}$, we let 
\begin{align*}
  \delta^\N_a((q_p, c_p, \varphi_{p})_{p\in X}) = ((q'', c'', \varphi''), \ldots, (q'', c'', \varphi''))
\end{align*}
 where $(q'', c'', \varphi'') = \cut(\expand(\synchronise((q_p, c_p, \varphi_{p})_{p\in X}), a))$; 
Finally, we let $F^\N$ be the set of states $(q_p, c_p, \varphi_{p})_{p\in P}$ such that 
$\synchronise((q_p, c_p, \varphi_{p})_{p\in P}) = (q', c', \varphi')$ and the state reached by $\A$ from $q'$ by reading all the letters of the steps of $\varphi'$, from left to right, is accepting in $\A$: once again, since $\A$ satisfies the diamond property, the order in which the letters of a step of $\varphi'$ are read does not matter.

\begin{lemma}\label{lem:Binfty-BN}
  The AA $\B^\infty$ and $\B^\N$ recognise the same language. 
\end{lemma}
\begin{proof}
  We show by induction on the trace $t$ that the local states $(\varphi_{p})_{p\in P}$ reached after reading $t$ in $\B^\infty$ and the local states $(q_p, c_p, \psi_{p})_{p\in P}$ reached after reading $t$ in $\B^\N$ are such that for all $p\in P$,
  \begin{equation}
    (q_p, c_p, \psi_{p}) = \cut(q_0, 0, \varphi_{p})
  \label{eq:Binfty-BN}
  \end{equation}
  We will use the result of Lemma~\ref{lem:Binfty} stating moreover that $\varphi_p = \Foata(\view_p(t))$.

  \eqref{eq:Binfty-BN} trivially holds at the beginning of the runs. Suppose that \eqref{eq:Binfty-BN} holds for a trace $t$. We show it for the trace $t'$ obtained by adding a new event $a$ at the end. Let $(\varphi_{p})_{p\in P}$ and $(q_p, c_p, \psi_{p})_{p\in P}$ be the respective local states reached after reading $t$ in $\B^\infty$ and $\B^\N$. Let $X = \loc(a)$. By induction hypothesis, $(q_p, c_p, \psi_{p}) = \cut(q_0, 0, \varphi_{p})$ for all $p\in P$. Since $\varphi_p = \Foata{(\view_p(t))}$, we know (Lemma~\ref{lem:infiniteB-synchronisation}) that $\synchronise((\varphi_p)_{p\in X}) = \Foata(\view_X(t))$. 

  The new local states $(\varphi'_{p})_{p\in P}$ and $(q'_p, c'_p, \psi'_{p})_{p\in P}$ obtained after reading $a$ are unchanged for processes $p\notin \loc(a)$. We thus only consider processes in $X$. 
  We have \[\varphi'_p = \expand(\Foata(\view_X(t)), a),\] and \[(q'_p, c'_p, \psi'_{p}) = \cut(\expand(\synchronise((\cut(q_0, 0, \varphi_{p}))_{p\in X}), a)).\] By definition of $\cut$, for all $p\in X$, $\cut(q_0, 0, \varphi_{p})$ is a suffix of $\view_p(t)$ of length at least $2k-2$. By Lemma~\ref{lem:everyoneKnowsBeyond2KFoataComponents}, letting $p$ be the process of $X$ having the maximal value of $c_p$, for all $i\leq \f(\view_p(t),2k-2)$, all processes $p'\in X$ have the full information on the step $i$ of $\Foata(t)$. Thus, $\synchronise((\cut(q_0, 0, \varphi_{p}))_{p\in X})$ is a suffix of $\view_X(t)$ that fulfils the condition to be the Foata component in the states of $\B^\N$. The operations $\expand$ performed in the two automata only depend on the last step of the FNF and thus $\expand(\synchronise((\cut(q_0, 0, \varphi_{p}))_{p\in X}), a)$ is a suffix of $\expand(\Foata(\view_X(t)), a)$. Since we simply perform a $\cut$ on top of that in $\B^\N$, we get, for all $p\in X$, $(q'_p, c'_p, \psi'_{p}) = \cut(q_0, 0, \varphi'_{p})$.
\end{proof}

As a corollary of Proposition~\ref{prop:infinite-AA}, we have that the infinite AA $\B^\N$ also recognises the language of $\A$.

\subsection{Modulo counting to obtain a finite AA}
\label{sec:final-AA}
\newcommand\modu{\mathrm{mod}}

We finally build a third AA, that will be \emph{finite}, obtained by bounding the counter values taken in $\B^\N$. Formally, we let $\B^\modu= ((Q^\modu_p)_{p \in P}, \Sigma, q^\modu_0, (\delta^\modu_a)_{a \in \Sigma}, F^\modu)$ where the exponent is there to remember that the counter is bounded by taking it modulo a constant, more precisely, modulo $2k$.

We thus let $Q^\modu_p$, for all processes $p\in P$, be the set of states $(q, c, \varphi)$ with $q\in Q$,  $c\in \{0, 1, \ldots, 2k-1\}$ and $\varphi$ as the Foata component in $Q^\N$. We let $q^\modu_0 = q^\N_0$. We then update the definition of the three functions $\synchronise$, $\expand$ and $\cut$ to define the transitions and final states of $B^\modu$. 

Starting from a letter $a$, with $X=\loc(a)$, and a state $(q_p, c_p, \varphi_{p})$ for all $p\in X$, we let $\synchronise((q_p, c_p, \varphi_{p})_{p\in X}) = (q', c', S')$ as follows. The idea is to mimick the definition of the function $\synchronise$ in $\B^\N$. Everything is easy but the beginning of the definition where we must choose a process $p'$ with $c_{p'} = \max(c_p)$ to get a process with the most recent information counterwise. Indeed, the counters are now counted modulo $2k$, and therefore the relative order is not a priori known: it is possible that a counter value $c_{p'}=0$ has more recent information than a counter value $c_p = 2k-1$. Thanks to Lemma~\ref{lem:pairOfProcsAtmostK-1Apart}, we will maintain by induction the property that the values $(c_p+|\varphi_p|)_{p\in X}$ (representing the length of the views up to a constant value) are in a range of at most $k$ consecutive values modulo $2k$. This means that we can find a range of at least $k$ such values that are not taken, and thus we are able to give an order between these values that is consistent with the actual lengths of views. Suppose thus that we have a value $c$ such that $(d_p = c_p+|\varphi_p|+c\mod 2k)_{p\in X}$ are in the range $\{0, 1, \ldots, k-1\}$. Then, consider the process $p'$ such that $d_{p'}-c = \max(d_p-c)$, so that it is indeed a consistent choice with the one made in $\B^\N$. We then let $q' = q_{p'}$, $c' = c_{p'}$. For all processes $p\in X$, if $\varphi_{p} = S_1\cdots S_m$, and $c_p \cong c_{p'}-\alpha \mod 2k$ with $0\leq \alpha \leq k-1$ (by assumption on $c_{p'}$), we let $\psi_p = S_{j+1} \cdots S_m$ with $S_1\cdots S_j$ containing $\alpha$ letters. Finally, we let $\varphi' = \bigcup_{p\in X}\psi_p$.
As before, we let $\expand((q', c', \varphi'), a) = (q', c', \varphi'\{a\})$, generalising the previous definition of $\expand$.

We finally apply the $\cut$ function to make it a state of $B^\modu$. For a tuple $(q, c, S_1\cdots S_{m})$ (that is not a local state since the Foata component is too long), if the function $\cut$ in $\B^\N$ is such that $\cut(q, c, S_1\cdots S_{m}) = (q', c', S_j\cdots S_{m})$, we now let $\cut(q, c, S_1\cdots S_{m}) = (q', c' \mod 2k, S_j\cdots S_{m})$: the only change is the value of the counter that we consider modulo $2k$. In particular, the state $q'$ is still taken as the unique state of $\A$ reached from $q$ by reading all letters of $S_1 \ldots S_{j-1}$ (in any order, step by step, since $\A$ satisfies the diamond property and all letters of a step are independent).

Thanks to these new definitions, for all letters $a\in \Sigma$, and local states $(q_p, c_p, \varphi_{p})_{p\in X}$, we let 
\[\delta^\modu_a((q_p, c_p, \varphi_{p})_{p\in X}) = ((q'', c'', \varphi''), \ldots, (q'', c'', \varphi''))\] 
where $(q'', c'', \varphi'') = \cut(\expand(\synchronise((q_p, c_p, \varphi_{p})_{p\in X}), a))$. We also let $F^\modu$ be the set of states $(q_p, c_p, \varphi_{p})_{p\in P}$ such that 
$\synchronise((q_p, c_p, \varphi_{p})_{p\in P}) = (q', c', \varphi')$ and the state reached by $\A$ from $q'$ by reading all the letters of $\varphi'$ is accepting in $\A$: once more, since $\A$ satisfies the diamond property, the order in which the letters of a step of $\varphi'$ are read does not matter.

By relating the runs of $\B^\modu$ with the runs of $\B^\N$ (and thus of $\B^\infty$), and by the property of $k$-fairness of $\A$, we indeed get by induction the crucial property that in all reachable states of $\B^\modu$, all counter values of all processes, when translated by the length of the suffix of view that is remembered, are in a range of $k$ consecutive values modulo $2k$.

\begin{lemma}
\label{lem:Bmod-invariant}
  The AA $\B^\N$ and $\B^\modu$ recognise the same language.
\end{lemma}
\begin{proof}
  We show by induction on the trace $t$ that if the local states reached after reading $t$ in $\B^\N$ are $(q_p, c_p, \varphi_p)_{p\in P}$, then the local states reached after reading $t$ in $\B^\modu$ are $(q_p, c_p \mod 2k-2, \varphi_p)_{p\in P}$, and are such that all values $\{c_p+|\varphi_p|\mod 2k \mid p\in P\}$ are included in a set of the form $\{c, c+1, c+2, \ldots, c+(k-1)\}$, with $c\in \{0, 1, \ldots, 2k-1\}$.

  The property holds trivially for the empty trace. If it holds for a trace $t$, we show that it holds for the trace $t'$ obtained by adding a new event $a$ at the end. Nothing has to be done for processes not in $\loc(a)$. For the others, by the induction hypothesis on the range of values $\{c_p+|\varphi_p|\mod 2k \mid p\in P\}$, the $\synchronise$ operation (as well, trivially, as the $\expand$ and $\cut$) indeed performs the same thing with respect to the Foata components and the states in $\B^\N$ and $\B^\modu$. It only remains to show that the property on the counters remain true after the transition. This is indeed the result of Lemma~\ref{lem:pairOfProcsAtmostK-1Apart}, since we know that the Foata components consist in a suffix of the view of each process (by the proof of Lemma~\ref{lem:Binfty-BN}), and that these views cannot differ in length more than $k-1$ steps pairwise.
\end{proof}

As a corollary, we obtain that $\B^\modu$ recognises the language of $\A$. This ends the proof of Theorem~\ref{thm:fair-AA} (the size of the automaton is described at the end of Section~\ref{sec:contrib-fair}).

\begin{remark}
  In our construction, every process maintains a state $q$ of the DFA $\A$ in its local state. This is the state reached by $\A$ on reading the steps forgotten so far. However, not every process needs to maintain this information. The information about this state is needed only to compute the acceptance condition. It is sufficient to designate one process for storing the state information. While computing the final state reached in the DFA, we could first synchronise all the processes and then compute the actual final state of $\A$ using the state information stored by the designated process.
  %
\end{remark}

\section{Proof of the lower bound}
\label{app:lower-bound}

\newcommand{\Path}{\operatorname{Path}}

In \cite{GenGim10}, a family of languages $\Path_n$ is defined for which it is shown that every AA has size at least $2^{k}$ states, with $k = \frac{n}{4}$. As it is, the language $\Path_n$ is not fair. However, the proof considers a subset $L_n \incl \Path_n$ ($L_n$ is called $L$ in the original proof) which contains $2^k$ traces, and there is a process $n$ which must end up in different local states on each of these traces. We first notice that $L_n$ is $n$-fair. Moreover, the DFA accepting $L_n$ has size $\mathcal{O}(n)$. This results in the following lower bound.

\lowerBound*
\begin{proof}
  The language $L_n$ is defined over the alphabet $\Sigma = {P\choose 2}$ where $P = \{1, \dots n\}$ is the set of processes with $\loc(\{i, j\}) = \{i, j\}$. Suppose $n = 4k$. We define for $0 \le m < k$, traces $a_m = \{4m, 4m+1\}\{4m+1, 4m+2\}\{4m+2, 4m+4\}$ and $b_m = \{4m, 4m+1\}\{4m, 4m+3\}\{4m+3, 4m+4\}$. We further define the set of letters $S_n=\{\{4m, n\} \mid 0 \le m < k\}$. The language $L_n$ is then defined by the regular expression $(a_0 + b_0)(a_1 + b_1)\dots (a_{k-1} + b_{k-1})S_n$.

  Since each component $(a_i+b_i)$ can be recognised by a DFA of $5$ states, the language $L_n$ is recognised by a DFA with $5k+2$ states, which is linear in $n$. Since each trace of $L_n$ is of length $3k+1$, it is trivially $(3k+2)$-fair and thus $n$-fair for $k > 1$.
  
  While the final component $S_n$ is not a part of the description of $L_n$ in Theorem 2 of \cite{GenGim10}, it is only a minor adjustment allowing us to follow the rest of the proof of Theorem 2 of \cite{GenGim10} to conclude that process $n$ must have at least $2^k$ local states, thus proving our theorem.
\end{proof}

\section{Recalling~\cite{KRISHNA2013109}}
\label{sec:krishna-muscholl}
 In this section, we recall the work of Siddharth Krishna and Muscholl \cite{KRISHNA2013109} where the restriction is imposed on the \emph{communication architecture} of the distributed alphabet.

\paragraph*{Acyclic communication.}
\label{sec:acyclic-recall}
Given a set $P$ of processes, and a distributed alphabet $(\Sigma, \loc)$, an undirected graph called the \emph{communication graph} $\gcomm$ is constructed as follows: vertices of $\gcomm$ are the processes in $P$; for $p_1, p_2 \in P$, there is an edge $(p_1, p_2)$ if $p_1$ and $p_2$ share a common action, i.e.~$\Sigma_{p_1} \cap \Sigma_{p_2}\neq \emptyset$. We first recall two restrictions considered in \cite{KRISHNA2013109}:
\begin{itemize}
\item actions are either binary or local to a single process: for every $a \in \Sigma$, we have $|\loc(a)| \le 2$,
 \item the communication graph $\gcomm$ is connected and acyclic, and hence is a tree. 
 \end{itemize}
 We fix an arbitrary process $\roo$ as the root of $\gcomm$. Then, for each process $p$, we can assign a parent in this tree, denoted as $\parent(p)$. Secondly, we denote by $X_p$, the set of processes in the subtree rooted at $p$ in $\gcomm$ (including $p$).
 
Starting with a DFA specification $\A = (Q, \Sigma, q_0, \delta, F)$ satisfying the $I$-diamond property, the main challenge in synthesising an AA, as always, is to abstract the infinite AA maintaining the views of each process, into a finite AA. The technique of truncating the views to a bounded suffix, as used for $k$-fair specifications, does not work here: for instance, the specification could allow an unbounded number of local actions between two synchronisations, and hence a bounded suffix of the view does not carry enough information. We need to be able to maintain sufficient information about the view of each process so that we are able to synchronise processes and also compute the global state of the trace reached, simply by looking at the tuple of local states (and thereby determine whether to accept or not). 

We recall an AA $\Bacy$ as described in~\cite{KRISHNA2013109}, which achieves these goals. Each process $p$ maintains a pair of states $(\overleftarrow{q_p}, q_p)$ of $\A$ so that: 
\begin{itemize}
  \item $\overleftarrow{q_p}$ is the state reached by the DFA $\A$ on reading $\jointview_p(t)$, which is the smallest ideal of $\view_p(t)$ containing all actions where $p$ synchronises with its parent in $\gcomm$, 
\item $q_p$ is the state reached by the DFA $\A$ on reading $\view_p(t)$.
\end{itemize}
The transitions are designed to maintain this invariant. When process $p$ at state $(\overleftarrow{q_{p}},q_{p})$ encounters a local action $a$, it moves to $(\overleftarrow{q_{p}}, q'_p)$ where $q'_p = \delta(q_p, a)$. When $p$ synchronises with $\parent(p)$, both processes need to first reconcile their views. Suppose $(\overleftarrow{q}, q)$ is the state of $\parent(p)$ and $(\overleftarrow{q_p}, q_p)$ is the state of $p$. The last time $p$ and its parent synchronised with each other, they reached state $\overleftarrow{q_p}$ (according to the above invariant). Since then, $p$ has seen a sequence of actions $u$ leading to $q_p$, and its parent has seen a sequence $v$ leading to $q$. Observe that $\loc(u) \incl X_p$ and $\loc(v) \cap X_p = \emptyset$. Therefore, $\loc(u)\cap \loc(v)=\emptyset$, and the state reached on reading $v$ from $q_p$ is the same as the state reached on reading $u$ from $q$ (Figure~\ref{fig:diamond}-left). However, we do not know what $u$ and $v$ are. Fortunately, we do not need it. As long as we consider words $u$ with $\loc(u) \incl X_p$ and words $v$ whose domains do not intersect $X_p$, the bottom state of the diamond is the same, as made precise by the next lemma, and illustrated in the middle of Figure~\ref{fig:diamond}.

\begin{figure}[t]
  \centering
\scalebox{0.85}{\begin{tikzpicture}
  \node (s) at (0,0) {\footnotesize $q_p$};
  \node (sp) at (1,1) {\footnotesize $\overleftarrow{q_p}$};
  \node (r) at (2, 0) {\footnotesize $q$};
  \draw [snake=coil, segment aspect = 0,->, >=stealth, red, thick] (sp) to (s);
  \draw [snake=coil, segment aspect = 0,->, >=stealth, blue, thick] (sp) to (r);
  \node at (0.3, 0.6) [red] {\scriptsize $u$};
  \node at (1.7, 0.6) [blue] {\scriptsize $v$};
  \node (t) at (1,-1.15) {\footnotesize $q'=\diam(\overleftarrow{q_p}, q_p, q, X_p)$};
  \draw [snake= coil, segment aspect = 0, ->, >=stealth, blue, thick] (s) to (t);
  \draw [snake = coil, segment aspect = 0, ->, >=stealth, red, thick] (r) to (t);
  \node at (0.3, -0.6) [blue] {\scriptsize $v$};
  \node at (1.7, -0.6) [red] {\scriptsize $u$};

  \draw [thin, gray] (2.8, 1.6) -- (2.8, -1.6);

  \begin{scope}[xshift=3.3cm]
    \draw [dashed] (0,-1.5) rectangle (4.5,1.5);
    \draw (0.5, -0.5) rectangle (2, 0.5);
    \node at (1.25,0) {\scriptsize $t_1$};
    \node[right] at (2, 0) {\scriptsize $q_1$};
    \draw (2.5, 0.5) rectangle (3.8, 1.3);
    \node at (2.8, 0.9) {\scriptsize $t_2$};
    \node [red] at (3.3, 1.1) {\tiny $\incl X$};
    \node[right] at (3.8, 0.9) {\scriptsize $q_2$};
    \draw (2.5, -0.5) rectangle (3.8, -1.3);
    \node at (2.8, -1) {\scriptsize $t_3$};
    \node [blue] at (3.3, -0.7) {\tiny $\incl P\setminus X$};
    \node[right] at (3.8, -0.9) {\scriptsize $q_3$};
    \node[right] at (2.2, -1.8) {\scriptsize $\diam(q_1, q_2, q_3, X)$};
  \end{scope}
  \draw [thin, gray] (8.3, 1.6) -- (8.3, -1.6);

  \begin{scope}[xshift=11cm, yshift=1cm, scale=0.8]
    \node at (0, -3.2) {\scriptsize Tree-of-bags architecture};
    \begin{scope}[every node/.style={circle, fill, inner sep=2pt}]
      \node (0) at (0,0) {};
      \node (1) at (-1,-1) {};
      \node (2) at (1,-1) {};
      \node (3) at (-1.5, -2) {};
      \node (4) at (-0.5, -2) {};
      \node (5) at (1.5, -2) {};
    \end{scope}

    \begin{scope}[every node/.style={rectangle, fill, blue, inner sep=2pt}]
      \node (a) at (-0.4, 0.5) {};
      \node (b) at (0.4, 0.5) {};
      \node (c) at (-1.4, -0.3) {};
      \node (d) at (-1.4, -0.8) {};
      \node (e) at (-1.7, -0.5) {};
      \node (f) at (1.2, -0.6) {};
      \node (g) at (1.5, -0.3) {};
      \node (h) at (1.8, -0.6) {};
      \node (i) at (1.5, -0.8) {};
      \node (j) at (-2.2, -1.7) {};
      \node (k) at (-2.5, -2) {};
      \node (l) at (-2.1, -2.2) {};
      \node (m) at (2, -1.8) {};
      \node (n) at (2, -2.2) {};
    \end{scope}

    \begin{scope}[thin, gray]
      \draw (0) to (a) to (b) to (0);
      \draw (0) to (1);
      \draw (0) to (2);
      \draw (1) to (d) to (c) to (e) to (d);
      \draw (1) to (3);
      \draw (1) to (4);
      \draw (2) to (5);
      \draw (2) to (f);
      \draw (f) to (g);
      \draw (f) to (i);
      \draw (i) to (g);
      \draw (i) to (h);
      \draw (3) to (j) to (k) to (l) to (j);
      \draw (3) to (k);
      \draw (5) to (m) to (n) to (5);
    \end{scope}

    \begin{scope}[very thin, red, dashed]
      \draw (0) .. controls (-2, 1) and (2,1) .. (0);
      \draw (1) .. controls (-3, -1) and (-1, 1) .. (1);
      \draw (2) .. controls (1, 1) and (3, -1) .. (2);
      \draw (3) .. controls (-3, -0.5) and (-3, -3.5) .. (3);
      \draw (5) .. controls (2.5, -0.5) and (2.5,-3.5) .. (5); 
    \end{scope}
  \end{scope}


\end{tikzpicture}}
\caption{Left: Illustration of the diamond property. Middle: illustration of the $\diam$ function. Right: A tree-of-bags architecture. The picture shows the communication graph. The red dashed bubbles represent the bags. Each bubble has an outer process shown as a black circle. The blue squares are the inner processes.}
\label{fig:diamond}
\end{figure}

\begin{lemma} \label{lem:diamond}
There is a function $\diam\colon Q \times Q \times Q \times 2^P \to Q$ satisfying the following. For all states $q_1, q_2, q_3 \in Q$, subsets of processes $X \incl P$, traces $t_1, t_2, t_3$ such that \begin{enumerate}
  \item $\delta(q_0, t_1) = q_1$, $\delta(q_0, t_1 t_2) = q_2$, $\delta(q_0, t_1 t_3) = q_3$, 
  \item and $\loc(t_1) \incl X$, $\loc(t_2) \incl P \setminus X$,
\end{enumerate}
 we have $\delta(q_0, t_1 t_2 t_3) = \diam(q_1, q_2, q_3, X)$. 
\end{lemma}

Therefore, when a process $p$ and its parent jointly read a letter $a$, they first reconcile their views by computing the state $q' = \diam(\overleftarrow{q_p}, q_p, q, X_p)$ and then computing the $a$-successor $\delta(q',a)$. Process $p$ then transitions to its local state $(\delta(q',a), \delta(q',a))$, the first component being indeed the state that was reached the last time it synchronised with its parent, while process $\parent(p)$ transitions to its local state $(\overleftarrow{q}, \delta(q',a))$ since the first component does not change.

Finally, one needs to determine the global states $((\overleftarrow{q_p}, q_p))_{p \in P}$ that are accepting in $\Bacy$. If $t$ is the trace that has been read by the AA leading to this global state, by design, we have $q_p = \delta(q_0, \view_p(t))$. Using Algorithm~\ref{alg:globalState}, we can compute $\delta(q_0, \view_{X_p}(t))$ --- the state reached on the union of views of processes in $X_p$. Correctness of the algorithm follows from Lemma~\ref{lem:diamond}. Hence, the state reached by $\A$ on $t$ is obtained by running the algorithm on the process $\roo$. 

\begin{algorithm}
  \caption{Global state function}\label{alg:globalState}
  \begin{algorithmic}[1]
    \Require  A process $p$, a global state $((\overleftarrow{q_p}, q_p))_{p \in P}$
    \Ensure $q = \delta(q_0, \view_{X_p}(t))$
    \Function{$\state(p, ((\overleftarrow{q_p}, q_p))_{p \in P})$}{}
    \State $q \gets q_p$
    \For{each child $p'$ of $p$}
    \State $q \gets \diam(\overleftarrow{q_{p'}}, \state(p'), q, X_{p'})$
    \EndFor
    \State \Return $q$
    \EndFunction
  \end{algorithmic}
\end{algorithm}

\begin{lemma}
  Let $t$ be a trace and let $((\overleftarrow{q_p}, q_p))_{p \in P}$ be the global state reached on $t$ by the AA $\Bacy$. Then $\delta(q_0, t) = \state(\roo, ((\overleftarrow{q_p}, q_p))_{p \in P})$.
\end{lemma}

As a corollary, accepting global states in $\Bacy$ are the ones such that the computation of $\state(\roo, ((\overleftarrow{q_p}, q_p))_{p \in P})$ above gives an accepting state of $\A$.

\section{The detailed construction for Section~\ref{sec:combination}.}
\label{app:full-construction-treeofbags}

The construction of the AA in this case essentially simulates the construction of Section~\ref{sec:contrib-fair} restricted to bags, and additionally performs a construction similar to~\cite{KRISHNA2013109} for the outer processes. The local states on reading an input trace $t$ are as follows. 

\begin{itemize}
\item For every inner process $\iota$ from a bag $B$, we maintain the (potentially infinite) trace $\view_\iota(\trest)$. Notice that $\trest$ is the trace $t$ restricted to the bag alphabet of $B$, and it is thus $k_B$-fair.
\item For every outer process $o$, we keep a pair $(\overleftarrow{q_o}, q_o)$ where $q_o= \delta(q_0, \view_o(t))$ is the state reached by $\A$ on reading $\view_o(t)$, and $\overleftarrow{q_o} = \delta(q_0, \jointview_o(t))$, where $\jointview_o(t)$ is the smallest ideal containing all actions where $o$ synchronises with $\parent(o)$. Along with this pair, we also maintain $\view_o(\trest)$, to simulate its behaviour as its inner counterparts. 
\end{itemize}
 

Therefore the inner processes are trying to mimick the construction of Section~\ref{sec:contrib-fair}, whereas the outer processes run both this construction and the acyclic-architecture construction in parallel. There is one challenge though: the trace $\trest$ is not an ideal of $t$, so the state $\delta(q_0, \trest)$ has no meaning in the global context of the given trace, and it does not help in computing $\delta(q_0, t)$. What we really need is the state reached on $\view_B(t)$, that is $\delta(q_0, \view_B(t))$, so that we can use the technique of \cite{KRISHNA2013109} (Algorithm~\ref{alg:globalState} in Appendix~\ref{sec:krishna-muscholl}) in order to compute $\delta(q_0, t)$. 
Algorithm~\ref{alg:bagview} shows how to compute the state $q_B =\delta(q_0, \view_B(t)) $ from the states maintained as above. It indeed takes the full trace $\view_\iota(\trest)$ as input. Using the fairness hypothesis on bags, we will discuss later how the algorithm can be implemented using a bounded suffix of this view (of length depending on $k_B$).

Here is an intuitive idea of the algorithm. What part of $\view_B(t)$ do we already have? We have access to $\view_o(t)$, in particular $q_o = \delta(q_0, \view_o(t))$. If we know the trace $t'_1$ obtained from $\view_B(t)$ by removing events in $\view_o(t)$, then we can get $\delta(q_0, \view_B(t))$ as $\delta(q_o, t'_1)$. Line~$1$ computes $\view_B(\trest)$. Line~$3$ computes trace $t'$ as the trace obtained by removing $\view_o(\trest)$ from $\view_B(\trest)$. 

The lemma below shows that $t'_1$ is equal to $t'$ --- trace~$t_2'$ of the lemma below is indeed the trace $t'$ used in Line~3.

\begin{lemma}\label{lem:algoBag}
Let $t$ be a trace, $B \in \BB$, and $o = o(B)$. The following two traces are equal:
\begin{itemize}
  \item the trace $t'_1$ obtained by removing the prefix $\view_o(t)$ from $\view_B(t)$.
  \item the trace $t'_2$ obtained by removing prefix $\view_o(\trest)$ from $\view_B(\trest)$. 
\end{itemize}
\end{lemma}
\begin{proof}
Let $f_1$ be the maximal $\Sigma^{\inn}(o)$ event in $\view_B(t)$ and $f_2$ be the maximal $\Sigma^{\out}(o)$ event in $\view_B(t)$.
  It is clear that any any $\Sigma \setminus \Sigma^{\inn}(B)$ event in $\view_B(t)$ is in $\view_o(t)$, and in particular in $\ideal{f_2}$, since $o$ is the only process that communicates with processes outside the bag. Thus there are no $\Sigma \setminus \Sigma^{\inn}(B)$ events in $t'_1$. We show there are no such events in $t'_2$ either. If $f_2 < f_1$ (where $\le$ is the order in the partial order representation of trace $t$), then $f_1 = \max_o(\trest)$ and the claim follows. If $f_1 < f_2$, then for any $\Sigma \setminus \Sigma^{\inn}(B)$ event $g$ such that $g \not< f_1$ but $g < f_2$ there is no $\Sigma^{\inn}(B)$ event $h$ such that $g < h$. For otherwise there must be a $\Sigma^{\inn}(o)$ event $h'$ such that $g < h' < h$ which contradicts the fact that $g \not< f_1$.
  
  Note that for any $\Sigma^{\inn}(B)$ event $e$, if $e < f_2$, then $e \le f_1$ since there must be a $\Sigma^{\inn}(o)$ event $h$ such that $ e \le h < f_2$ and $h \le f_1$. Consequently, $t'_1$ and $t'_2$ contain the same events. It only remains to show that the order is the same as well.

  We now show that any two events $e, f$ such that $\loc(e), \loc(f) \subseteq \Sigma^{\inn}(B)$ are ordered in $\view_B(\trest)$ if and only if they are ordered in $\view_B(t)$. If $e$ and $f$ are ordered in $\view_B(\trest)$ then they are obviously ordered in $\view_B(t)$. We prove the other direction. Let $e = g_1 \lessdot \cdots \lessdot g_n = f$ be a sequence of events in $\view_B(t)$, and $g_\alpha$ be the last $\Sigma^{\inn}(B)$ event and $g_\beta$ be the last $\Sigma \setminus \Sigma^{\inn}(B)$ event. We then have $o \in \loc(g_\alpha)$ and $o \in \loc(g_{\beta + 1})$. Since both $g_\alpha$ and $g_{\beta + 1}$ are $\Sigma^{\inn}(B)$ events, we have $e < f$ in $\view_B(\trest)$.
\end{proof}

\viewB*
\begin{proof}
Directly follows from Lemma~\ref{lem:algoBag}.
\end{proof}

Once we know how to compute $q_B$ for each bag, we can extend Algorithm~\ref{alg:globalState} (Appendix~\ref{sec:acyclic-recall}) to the bag setting in a similar manner. We describe the new procedure in Algorithm~\ref{alg:globalStateCombo}. 

\smallskip 
\noindent \textbf{Making the construction effective.}
We cannot maintain $\view_\iota(\trest)$ in full as this is infinite. However, all we need is to be able to calculate the trace $t'$ of Algorithm~\ref{alg:bagview}. Since process $o(B)$ does not participate in $t'$, we have $|t'| \le k - 1$. Secondly the event $e$ of Algorithm~\ref{alg:bagview} also appears in the last $k$ letters of $\view_B(\trest)$. This shows that both $e$ and $t'$ are present in the last $k$ letters of $\view_B(\trest)$. Now, we can almost take the construction of Section~\ref{sec:contrib-fair} and apply it bagwise. Taken off the shelf, the construction maintains for each $\iota \in B$, a state of $\A$, a counter modulo $2k_B$ and a suffix of $\view_\iota(\trest)$ with $2k_B-2$ letters. As said earlier, the state does not make sense. We do not need it either --- we maintain only the last two parts, the counter value and the suffix. With this information in all the processes of $B$, we can apply Algorithm~\ref{alg:bagview} to compute $e$ and $t'$. The outer processes maintain this information coming from the fairness construction, and also a pair of states of~$\A$ for the acyclic construction.

Given a distributed alphabet $(\Sigma, \loc)$ that forms a tree-of-bags architecture, and a DFA $\A$ that is fair for $(\Sigma,\loc)$, we provide a construction to synthesise an asynchronous automaton $\Bcomb$ with the same language. Intuitively, inner processes apply the construction of Section~\ref{sec:final-AA} restricted to letters in their bag, while outer processes maintain two things: the construction of Section~\ref{sec:final-AA} restricted to their bag and the construction for acyclic communication recalled in the beginning of Section~\ref{sec:acyclic-recall}. 
\begin{description}
\item[Inner process.] Consider an inner process $\iota$. Let $B$ be the bag containing $\iota$, and let $\Sigma_B$ be the union of the alphabets of all \emph{inner processes} in $B$. Notice that $\Sigma_B$ does not include the letters that the outer process of $B$ uses to synchronise with other outer processes. On reading a trace $t$, process $\iota$ maintains a pair $(c_\iota, \varphi_\iota)$ such that: 
  \begin{align*}
    (c_\iota, \varphi_\iota) = \cut(0, \view_\iota(\trest))
  \end{align*}
  The arguments for the $\cut$ function above do not include the state information. 
  
  \item[Outer process.] Consider an outer process $\theta$. Let $B$ be its bag, and let $\Sigma_B$ be as above. Notice that $\Sigma_B$ also contains letters that $\theta$ shares with inner processes in its bag. On reading trace $t$, process $\theta$ maintains:
  \begin{align*}
    (\overleftarrow{q_\theta}, q_\theta), (c_\theta, \varphi_\theta) 
  \end{align*}
  where:
  \begin{itemize}
    \item $\overleftarrow{q_\theta}$ is the state reached by $\A$ on reading $\jointview_\theta(t)$, which is the smallest prefix of $t$ containing all actions where $\theta$ synchronises with $\parent(\theta)$,
    \item $q_\theta$ is the state reached by $\A$ on reading $\view_\theta(t)$
    \item $(c_\theta, \varphi_\theta)$ is such that $\cut(0, \view_\theta(\trest))$ is of the form $(q_\theta, c_\theta, \varphi_\theta)$
  \end{itemize}
  We will call $(\overleftarrow{q_\theta}, q_\theta)$ as the outer component of the local state, and $(c_\theta, \varphi_\theta)$ as the inner component.
\end{description}

The transitions of the inner processes and the inner component of outer processes is the same as the fairness construction: synchronise, expand and cut. Let us explain the transitions on the outer component. Suppose $(\overleftarrow{q_\theta}, q_\theta)$ is the state of outer process $\theta$. Let $(\overleftarrow{r_\theta}, r_\theta)$ be the state of $\parent(\theta)$. Suppose $\theta$ and $\parent(\theta)$ synchronise on a joint action $a$. Let $X_\theta$ be the union over all bags $B$ such that the outer process $o(B)$ is in the substree containing $\theta$ (including $\theta$) -- therefore, $X_\theta$ contains the process $\theta$, its bag, and all the bags in the subtree of $\theta$. To first reconcile the views, the processes compute $q' =  \diam(q_\theta, q'_\theta, r_\theta, X_\theta)$. Then both $\theta$ moves to $(\delta(q', a), \delta(q', a))$ and $\parent(\theta)$ moves to $(\overleftarrow{r_\theta}, \delta(q', a))$.

For a bag $B$, let~$X_B$ be the union of all the bags that include $B$ and all the bags in the subtree of $B$.
The correctness of this algorithm follows from the next lemma.

\begin{algorithm}
  \caption{Global state function}\label{alg:globalStateCombo}
  \begin{algorithmic}[1]
    \Require  A bag $B$, a tuple of states $((\overleftarrow{q_o}, q_o))$ for all outer processes in $B$ 
    \Ensure $q = \delta(q_0, \view_{X_B}(t))$
    \Function{$\cstate(B)$}{}
    \State $q \gets \delta(q_0, \view_B(t))$ \Comment using Algorithm~\ref{alg:bagview}
    \For{each child $B'$ of $B$}
    \State $q \gets \diam(\overleftarrow{q_{\oo(B')}}, \cstate(B'), q, X_{B'})$ 
    \EndFor
    \State \Return $q$
    \EndFunction
  \end{algorithmic}
\end{algorithm}


\begin{restatable}{lemma}{cState}
  Let $t$ be a trace and let $\overleftarrow{q_o}$ and $q_o$ be the states reached by $\A$ respectively on $\jointview_o(t)$ and $\view_o(t)$, for every outer processes $o$. Then $\delta(q_0, t) = \cstate(B_{\roo}, ((\overleftarrow{q_o}, q_o))_{o})$.
\end{restatable}
\begin{proof}
The proof closely follows the proof of Proposition 4.2 in \cite{KRISHNA2013109}. We choose an enumeration of the children of $B$, $B_1, \dots B_m$. We prove by induction on the number of children that the for-loop of lines 3--4 maintains the invariant $q = \view_{\{B\} \cup X_{B_1} \ldots \cup X_{B_i}}(t)$ after iteration $i$ for $i \in \{1, \dots, m\}$. We also recursively assume that $\cstate$ is correct for all the bags in the subtree of $B$, excluding $B$. The invariant is satisfied at line 2. We assume the invariant holds before iteration $i$ and show that it holds after it.

Let $o_i = \oo(B_i)$. We decompose the trace $\view_{\{B\} \cup X_{B_1} \ldots \cup X_{B_i}}(t) = t_0t_1t_2$ where $t_0 = \jointview_{o'}(t)$, $t_0t_1 = \view_{B_i}(t)$ and $t_0t_2 = \view_{\{B\} \cup X_{B_1} \ldots \cup X_{B_{i-1}}}(t)$. Moreover $\delta(q, t_0) = \overleftarrow{q_{o'}}$, $\delta(q, t_0t_2) = q$ and $\delta(q_0, t_0t_1) = \cstate(B_i)$ by the inductive hypothesis. Since $\loc(t_1) \subseteq X_{B_i}$ and $\loc(t_2) \cap X_{B_i} = \emptyset$ we have by \Cref{lem:diamond}, $\delta(q_0, t_0t_1t_2) = \diam(\overleftarrow{s_{o_i}}, \cstate(B_i), q, X_{B_i})$. 
\end{proof}

\section{Use of the tool on a small example}
\label{app:tool}

  As an example of usage of the tool, consider the DFA of Figure~\ref{fig:DFA-unfair} where we again consider two processes $p_1$ and $p_2$, and three letters: $a$ that is local to $p_1$, $b$ that is local to $p_2$ and $c$ that is shared by both processes. This DFA accepts all traces where every two consecutive occurrences of $c$ (or an occurrence of $c$ and the beginning/end of the trace) is separated by an even number of letters $a$ or $b$. This is unfair since the word $(aa)^k$ is accepted for all $k$, whereas process $p_2$ is not participating. For a fairness parameter $k$, we can build the AA obtained via our construction to recognise exactly all traces accepted by the DFA that are $k$-fair. For instance, if $k=1$, we only accept traces of the form $c^*$ since no letter $a$ or $b$ can be read. The generated AA has 3 local states for each process: $(0, 0, \{\})$, $(0, 0, \{c\})$, and $(0, 1, \{c\})$. The semantics of this AA is depicted (as produced by the tool) in Figure~\ref{fig:aa1}. For $k=2$, the AA should now allow for exactly one occurrence of $a$ and $b$ between two successive occurrences of $c$, which generates an AA where every process has $15$ local states: $(0, 0, \{\})$, $(0, 0, \{b\})$, $(0, 0, \{c\})$, $(0, 0, \{c\}\{b\})$, $(0, 0, \{c\}\{c\})$, $(0, 0, \{a,b\}\{c\})$, $(0, 1, \{c\}\{b\})$, $(0, 1, \{c\}\{c\})$, $(0, 2, \{c\}\{b\})$, $(0, 2, \{c\}\{c\})$, $(0, 1, \{a,b\}\{c\})$, $(0, 3, \{c\}\{b\})$, $(0, 3, \{c\}\{c\})$, $(0, 2, \{a,b\}\{c\})$, $(0, 3, \{a,b\}\{c\})$. The semantics of this AA is depicted (as produced by the tool) in Figure~\ref{fig:aa2}. Finally, if we use our optimised version of the algorithm, this results in an AA where every process has 9 local states: $(0, 0, \{\})$, $(0, 0, \{b\})$, $(0, 0, \{c\})$, $(0, 1, \{b\})$, $(0, 1, \{c\})$, $(0, 2, \{c\})$, $(0, 2, \{b\})$, $(0, 3, \{b\})$, and $(0, 3, \{c\})$. The semantics of this AA is depicted (as produced by the tool) in Figure~\ref{fig:aa2_optim}.

  \begin{figure}[tbp]
    \centering 
    \begin{tikzpicture}[state/.style={draw, circle, thick, inner
        sep=3pt},initial text=,>=latex]
      \node[state,accepting,initial] (0) at (0,0) {\scriptsize $0$};
      \node[state] (1) at (2, 0) {\scriptsize $1$};

      \draw[->] (0) edge[loop above] node{$c$} (0)
      (0) edge[bend left] node[above] {$a, b$} (1)
      (1) edge[bend left] node[below] {$a, b$} (0);
    \end{tikzpicture}
    \caption{An unfair trace-closed DFA\label{fig:DFA-unfair}}
  \end{figure}
  \begin{figure}[tbp]
    \centering
    \includegraphics[width=\linewidth]{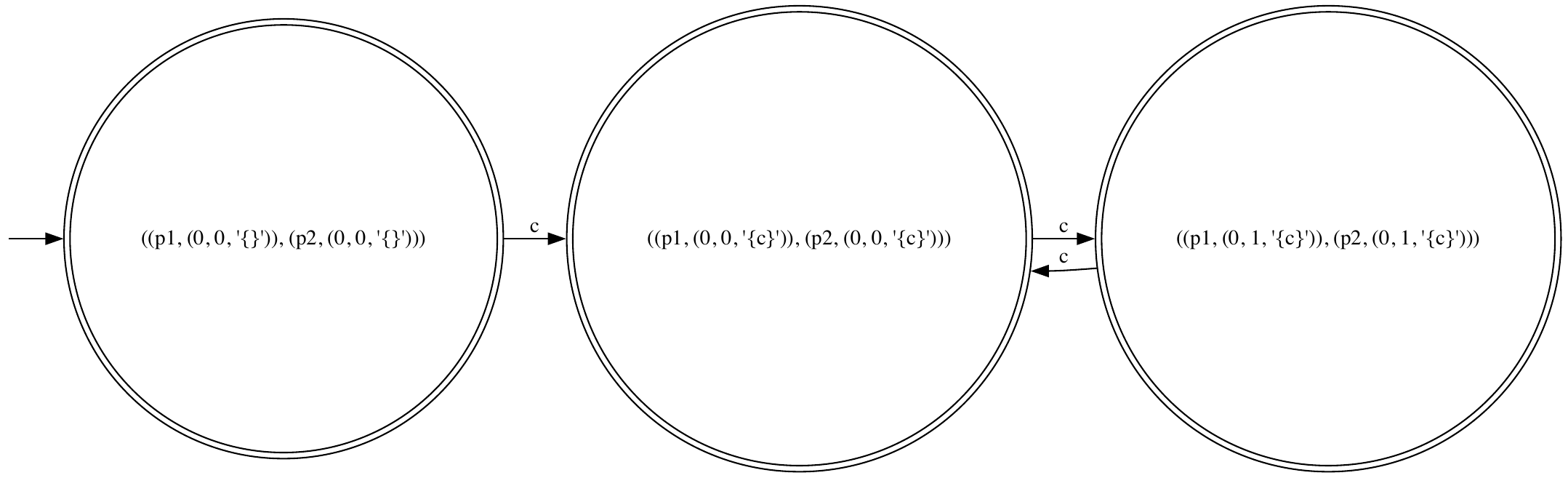}
    \caption{Semantics of the AA synthesised from the unfair DFA of Figure~\ref{fig:DFA-unfair} with fairness parameter $k=1$}\label{fig:aa1}
  \end{figure}
  \begin{figure}[tbp]
    \centering
    \includegraphics[width=\linewidth]{aa2.pdf}
    \caption{Semantics of the AA synthesised from the unfair DFA of Figure~\ref{fig:DFA-unfair} with fairness parameter $k=2$}\label{fig:aa2}
  \end{figure}
  \begin{figure}[tbp]
    \centering
    \includegraphics[width=\linewidth]{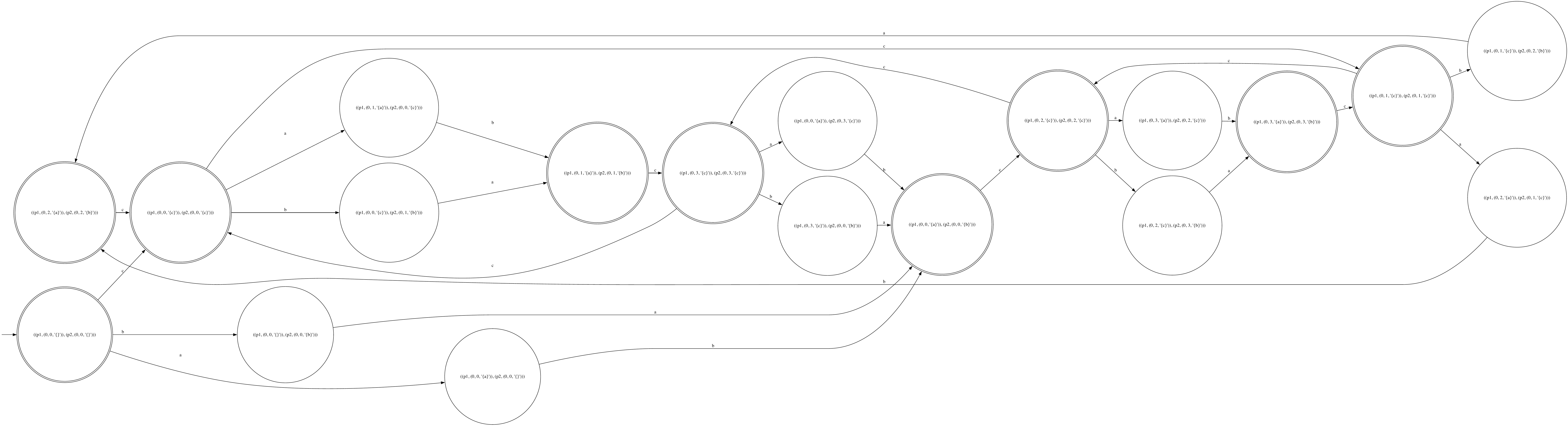}
    \caption{Semantics of the AA synthesised from the unfair DFA of Figure~\ref{fig:DFA-unfair} with fairness parameter $k=2$, with the optimised construction}\label{fig:aa2_optim}
  \end{figure}

\end{document}